\documentclass[12pt]{article}

\usepackage{setspace} 
\usepackage[vmargin = 1.25in, hmargin =1.25in]{geometry}

\usepackage[dvipsnames,svgnames]{xcolor}
\usepackage{amsmath, amssymb, amsfonts, graphicx, tikz,pdflscape, mathtools, amsthm, upgreek, bm,pgfplots,csvsimple,multirow, multicol, booktabs,bbm,tabularx,array}
\usepackage{verbatim}
\usepackage[longnamesfirst]{natbib}
\setcitestyle{citesep={,}}
\usepackage{accents}

\usepackage{needspace}
\def\citeapos#1{\citeauthor{#1}'s (\citeyear{#1})}
\usepackage{comment}
\usepackage[inline]{enumitem}
\usepackage{tocvsec2}
\usepackage{threeparttable}
\usetikzlibrary{calc, cd, arrows}
\usetikzlibrary{decorations.pathreplacing}
\usetikzlibrary{positioning}
\usetikzlibrary{calc} 
\usetikzlibrary{arrows}
\usetikzlibrary{patterns}
\usetikzlibrary{decorations.markings}
\usetikzlibrary{shapes.misc}
\usetikzlibrary{matrix,shapes,arrows,fit,tikzmark}
\usetikzlibrary{cd}
\usepgfplotslibrary{fillbetween}
\usetikzlibrary{intersections}
\usepackage[skip = 10pt]{caption}
\usepackage{subcaption}
\allowdisplaybreaks[1]
\allowdisplaybreaks[1]

\definecolor{Blue}{RGB}{86,180,233}
\definecolor{Orange}{RGB}{230,159,0}
\definecolor{Green}{RGB}{0,158,115}
\definecolor{GmailBlue}{RGB}{42, 93, 176} 
\usepackage[
	pagebackref,
	colorlinks=true,
	citecolor= GmailBlue,
	linkcolor=GmailBlue,
	urlcolor = GmailBlue
]{hyperref}

\newcommand{\bibtexorder}[1]{}

\usepackage{pgfplots}
\usepgfplotslibrary{groupplots,colorbrewer}
\pgfplotsset{compat=newest}
\pgfplotsset{cycle list/Set1}
\usepackage{tikz}
\usetikzlibrary{matrix,calc,shapes,arrows.meta,positioning}
\tikzset{
    vertex/.style = {shape=circle,draw, minimum size = 1.8em, inner sep = 0pt},
    edge/.style = {->,> = latex}
}

\usepackage[capitalize,noabbrev]{cleveref}

\makeatletter
\AddToHook{cmd/appendix/before}{\def\cref@section@alias{appendix}\def\cref@subsection@alias{appendix}}
\makeatother

\newtheoremstyle{break}
{}
{}
{\itshape}
{}
{\bfseries}
{}
{\newline}
{}

\theoremstyle{break}
\newtheorem{thm}{Theorem}
\newtheorem*{theorem*}{Theorem}
\newtheorem*{cor*}{Corollary}
\newtheorem{cor}{Corollary}
\newtheorem{prop}{Proposition}
\newtheorem{lem}{Lemma}

\crefname{prop}{Proposition}{Propositions}
\crefname{thm}{Theorem}{Theorems}
\crefname{lem}{Lemma}{Lemmas}
\crefname{blem}{Lemma}{Lemmas}

\theoremstyle{definition}
\newtheorem{defn}{Definition}

\newtheorem{rem}{Remark}
\newtheorem*{rem*}{Remark}
\newtheorem*{claim*}{Claim}

\def\a{\alpha}

\def\d{\delta}
\def\e{\varepsilon}

\def\h{\eta}
\def\th{\theta}

\def\k{\kappa}
\def\l{\lambda}

\def\s{\sigma}

\def\x{\chi}

\def\D{\Delta}
\def\Th{\Theta}

\def\N{\mathbf{N}}

\def\R{\mathbf{R}}

\def\HH{\mathcal{H}}

\DeclareMathOperator{\supp}{supp} 
\DeclareMathOperator*{\argmax}{argmax}

\DeclareMathOperator{\dev}{dev}
\DeclareMathOperator{\ob}{ob}
\DeclareMathOperator{\marg}{marg} 

\newcommand{\Abs}[1]{\left\lvert #1 \right\rvert}

\newcommand{\Brac}[1]{\left[ #1 \right]}

\newcommand{\Set}[1]{\left\{ #1 \right\}}

\usepackage{datetime}
\newdateformat{specialdate}{\THEDAY~\monthname[\THEMONTH] \THEYEAR}

\title{The Robustness of Sequential Implementation under Information Perturbations\thanks{We thank Pierpaolo Battigalli and Takashi Kunimoto  for  helpful comments. For financial support, we thank the National Science Foundation  grant SES-2417162 and the Deutsche Forschungsgemeinschaft (DFG, German Research Foundation) under Germany's Excellence Strategy – EXC-2047/2 – 390685813.}}
\author{Ian Ball\thanks{Department of Economics, MIT, \texttt{ianball@mit.edu}.} \and Drew Fudenberg\thanks{Department of Economics, MIT, \texttt{drewf@mit.edu}.} \and Stephen Morris\thanks{Department of Economics, MIT, \texttt{semorris@mit.edu}.}}
\date{\specialdate\today}

\begin{document}

\maketitle

\begin{abstract}
Bayesian monotonicity is a necessary condition for full implementation in Bayes--Nash equilibrium. Under the sequential equilibrium refinement, however, multi-stage mechanisms can expand the scope for implementation. We show that this additional implementation power is fragile. Whenever a multi-stage mechanism sequentially implements a social choice correspondence that is not Bayesian monotone, we construct arbitrarily small ``information perturbations'' and show that there are sequential equilibria under these perturbations that yield inadmissible decisions with probability bounded away from zero. Under these perturbations, players' types are not independent conditional on the state and some player is sometimes uncertain of their payoff type.  We show by example that the general result fails if perturbations are required to satisfy conditional independence or ``known payoff types.''
\end{abstract}

\newpage

\section{Introduction}

The incomplete-information implementation problem asks which social choice correspondences can be fully implemented in pure Bayes--Nash equilibrium when players have privately observed payoff-relevant characteristics, which we call \emph{payoff types}.  
Multi-stage mechanisms and the sequential equilibrium refinement can expand the
scope for implementation relative to static mechanisms and Bayes--Nash equilibrium.
We show that this additional implementation power is not robust to arbitrarily small ``information perturbations'' that change the information available to players. In these perturbations, players observe private signals rather than
directly observing their payoff types. The signals can be arbitrarily informative
about the payoff types, and the distribution of payoff types is unchanged. We also give an example in which the additional implementation power of sequential equilibrium is robust to all sufficiently small perturbations that satisfy an additional property---either ``known payoff types'' or conditional independence. (These properties are described below and formally defined in Section~\ref{sec:KPT}.)

We consider a social choice environment with a common prior. Each player knows their own payoff type, but they do not know the payoff types of their opponents. We allow for correlated payoff types and interdependent values. Suppose that a multi-stage mechanism sequentially implements a social choice correspondence that  is not Bayesian monotone, and so cannot be fully implemented in pure Bayes--Nash equilibrium \citep{Jackson1991}.   Our main result, \cref{res:fragile}, constructs arbitrarily small information perturbations and pure sequential equilibria under these perturbations that induce decisions outside the social choice correspondence with probability bounded away from zero. This is an incomplete-information analogue of the complete-information result in \cite{AghionEtAl2012}. Unlike the complete-information construction, ours introduces richer signal spaces to coordinate
behavior across payoff-type profiles while preserving the players'
potentially nondegenerate off-path beliefs.

The intuition for our result is as follows. Consider a multi-stage mechanism that fully implements a social choice correspondence $F$ in pure sequential equilibrium. If $F$ is not Bayesian monotone,  the mechanism has a ``good'' sequential equilibrium that induces a social choice function in $F$ and a ``bad'' Bayes--Nash equilibrium (BNE) that induces a social choice function outside $F$. By the definition of full implementation, this bad BNE cannot be supported as a sequential equilibrium of the mechanism. We construct a family of perturbed information structures that converge to the baseline environment (in a sense that we formalize). In each perturbed information structure, each player observes a private signal rather than their true payoff type. We construct a pure sequential equilibrium in which the players coordinate on the bad BNE with high probability; otherwise, they coordinate on the good sequential equilibrium. From their signal realization, each player can learn their own payoff type with high confidence, but at least one player does not always learn their payoff type perfectly. If a player observes a deviation by one of their opponents, they infer that their opponent observed a particular revealing signal realization. They conclude that the players are actually coordinating on the moves specified by the good sequential equilibrium under the true payoff-type profile. Given this inference, sequential optimality follows. 

In our constructed information structures, players do not always know their own payoff types, and players' signals are correlated conditional on the true payoff-type profile. We show by example that the general result can fail if either
property is ruled out.  To do this formally, we consider two restricted classes of information structures. An information structure satisfies \emph{known payoff types (KPT)} if each player always knows their own payoff type. An information structure satisfies \emph{conditional independence (CI)} if the players' types are statistically independent, conditional on the true payoff-type profile. For each of these restricted perturbation classes, we construct an example of a social choice function that is not Bayesian monotone but can be sequentially implemented by a multi-stage mechanism in a way that is robust to all sufficiently small information perturbations in the restricted class. That is, the probability of an inadmissible outcome vanishes with the size of the perturbation under every pure sequential equilibrium. Formally, these examples show that the general non-robustness conclusion in \cref{res:fragile} no longer holds if the information perturbations are restricted to either the KPT or CI classes. 

The rest of the paper is organized as follows. \cref{sec:lit} discusses related literature. \cref{sec:example} introduces a simple example that illustrates the main construction. \cref{sec:setting} introduces the general social choice environment and standard implementation notions. \cref{sec:perturbations} defines information perturbations. \cref{sec:fragility} presents the main non-robustness result. In \cref{sec:KPT}, we show that the general non-robustness conclusion fails if perturbations are required to satisfy either KPT or CI. \cref{sec:discussion} concludes. \cref{sec:proof_lemmas} contains the proofs.

\subsection{Related literature} \label{sec:lit}

Under complete information, Maskin monotonicity is necessary for Nash
implementation \citep{Maskin1999}.\footnote{\citet{Maskin1999} also shows
that, with three or more players and under a no-veto-power condition,
Maskin monotonicity is sufficient for Nash implementation.}
\citet{MooreRepullo1988} show that multistage mechanisms can implement
social choice functions that violate Maskin monotonicity in subgame-perfect
equilibrium. However, \citet{AghionEtAl2012} show  that mechanisms
implementing such functions are vulnerable to arbitrarily small
information perturbations, in the spirit of
\citet{FKL1988}.\footnote{For the specific mechanisms in
\citet{MooreRepullo1988}, \citet{AghionEtAl2012} also show that truthful
equilibria can fail to survive small information perturbations.}
\citet{PalfreySrivastava1991} show that requiring undominated Nash
equilibrium can also expand the class of implementable social choice
functions, but \citet{ChungEly2003} show that this additional
implementation power is not robust to information perturbations.

With incomplete information, Bayesian monotonicity is a necessary condition for full implementation in pure Bayes--Nash equilibrium \citep{Jackson1991}. Subsequently, it was shown that under extensive-form refinements, multi-stage mechanisms can fully implement some  social choice functions that violate Bayesian monotonicity. \cite{Brusco1995} introduces necessary and sufficient conditions for full implementation in perfect Bayesian equilibrium. \cite{BerginSen1998} and \cite{Baliga1999} provide sufficient conditions for full implementation in pure sequential equilibrium.\footnote{In a quasilinear setting, \cite{Duggan1998} showss that incentive compatibility is sufficient for extensive-form implementation. \cite{Brusco2006} characterizes full implementation in perfect Bayesian equilibrium by a class of two-stage mechanisms with a public signal after the first stage. \citet{takahashi2020robust} studies mechanisms that are robust
to \emph{all}  information perturbations, and relaxes the convergence of strategies to the convergence of outcomes.} Our paper shows that this additional implementation power is not robust to small information perturbations.



Under incomplete information, the question arises whether to restrict to   information perturbations with the ``known-payoff-type''  (KPT) property. In a two-player, complete-information setting in which the state is a vector of payoff types, \cite{EcclesWegner2016} showsthat sequential implementation in two-stage mechanisms is robust to KPT perturbations, but sequential implementation in three-stage mechanisms may not be. Also starting from complete information, \citet*{ChenHoldenKunimotoSunWilkening2023}
construct two-stage mechanisms that implement any social choice function
under initial rationalizability and remain robust to small private-value
information perturbations, and  \citet*{chen2023continuous} show that the KPT assumption (and an additional assumption that rules out independent private values) makes a difference for the robustness of continuous implementation.  We do not provide a general positive robustness result for KPT perturbations, but we show that our general non-robustness result does not hold under a restriction to KPT perturbations.

A broader related literature  following  \citet{FKL1988} studies the robustness of equilibrium refinements
under elaborations of a benchmark game. In complete-information normal-form games, \citet{DekelFudenberg1990} characterize $S^\infty W$---one round of deletion of weakly dominated strategies followed by iterated deletion of strongly dominated
strategies---using elaborations with vanishing payoff uncertainty, while
\citet{Borgers1994} gives a related epistemic characterization in terms of
approximate common certainty of admissibility.  \citet{FrickRomm2015} showsthat when extending these results to incomplete information, the conclusion depends on whether
the baseline types retain exactly their original beliefs about payoff-relevant
states or are instead allowed arbitrarily small perturbations of those beliefs. In
the latter case, the set of predicted behavior is generally strictly larger. Our
result identifies a related sensitivity for sequential implementation: preserving
the ex-ante distribution of payoff types does not by itself preserve the baseline
sequential-equilibrium prediction, because allowing players to be uncertain about
their own payoff-type components can change the off-path inferences that support
sequential behavior.

Another strand of the literature focuses on refinements of rationalizability rather than Nash equilibrium. By the structure theorem in \cite{WeinsteinYildiz2007}, under a richness condition, refinements of rationalizability cannot yield robust predictions sharper than the predictions of rationalizability itself. Under incomplete information, \cite{Muller2016} shows that many social choice functions that are not virtually implementable under rationalizability can nevertheless be virtually implemented under strong rationalizability \citep{BattigalliSiniscalchi1999,BattigalliSiniscalchi2003}, which imposes a form of  forward induction.  \citet{BattigalliSiniscalchi2007} showsthat strong rationalizability is not robust to some forms of belief perturbations, while \cite{battigalli2026monotonicity} showsthat full implementation with respect to strong rationalizability   is robust to the planner lacking knowledge of the type space that represented the agents' possible belief hierarchies.  \cite{PiermontZuazoGarin2026} and \cite{Wang2026} analyze the robustness of strong rationalizability to particular classes of perturbations.

\section{Motivating example} \label{sec:example}

We begin with a simple example that illustrates  the additional power of implementation in pure sequential equilibrium and the non-robustness of such implementation. There are two players, denoted $i= 1,2$. Each player $i$ observes their payoff type $\th_i \in \Th_i = \{ 0,1\}$, which is drawn uniformly and independently. The set of decisions is 
\[
    A = \Set{ x^{\th', c, d}: \th' \in \Th~\text{and}~c,d \in \{0,1\}^2 }.
\]
Utility functions will be specified below. Consider the following three-stage mechanism. The structure of reporting, challenging, and defending is inspired by an example in \cite{MooreRepullo1988}.
\begin{enumerate}
    \item Each player $i$ reports $\th_i' \in \Th_i$. Play proceeds to stage 2. 
    \item After observing the reports in stage $1$, each player $i$ chooses whether to challenge ($c_i =1$) or not ($c_i = 0$). If neither player challenges, the game ends and the decision is $x^{\th'}\coloneq x^{\th', 0, 0}$. Otherwise, play proceeds to stage 3.
    \item After observing the challenge decisions in stage $2$, each player who was challenged chooses whether to defend ($d_i =1 $) or not ($d_i = 0$). The game ends and the decision is $x^{ \th', c, d}$.\footnote{We adopt the convention that $d_i = 0$ whenever $c_{-i} = 0$.} 
\end{enumerate}

The players have private values:
\begin{equation} \label{eq:u_decomposition}
    u_i ( x^{\th', c, d}, \th_i) = c_{-i} \cdot  v_i (\th_i', d_i, \th_i)  + c_i \cdot w_i(d_{-i}), 
\end{equation}
where 
\[
    v_i (\th_i', d_i, \th_i) =
    \begin{cases}
        -1 + d_i &\text{if}~\th_i = \th_i', \\
        -3 - d_i &\text{if}~\th_i \neq \th_i',
    \end{cases}
\]
and 
\[
    w_i(d_{-i}) = 2 - 3 d_{-i}.
\]
Player $i$'s payoff in \eqref{eq:u_decomposition} has two components: $v_i$ captures the effect of being challenged, and $w_i$ captures the effect of challenging. Upon being challenged, it is strictly optimal for a player to defend if they reported truthfully and not to defend if they misreported. Challenging their opponent changes the player's utility by $2$ if the opponent does not defend but by $-1$ if the opponent defends. 

\subsection{Equilibria} \label{sec:equilibria} First we consider pure sequential equilibria of the mechanism. We claim that the unique sequential equilibrium strategy profile is as follows. In stage $1$, each player reports truthfully. In stage $2$, neither player challenges. In stage $3$, each player who was challenged in stage $2$ defends if and only if they reported truthfully in stage $1$. To verify this claim, we work backwards. The sequentially optimal behavior in stage $3$ is pinned down, no matter what each player believes about their opponent's payoff type. In stage $2$, Bayes' rule implies that after each report profile, each player is certain that their opponent reported truthfully and hence will defend against a challenge, so challenging is strictly suboptimal. Finally, in stage 1, players are indifferent over all report profiles, so in particular, truthtelling is optimal. Consistency can be established from a suitable sequence of totally mixed strategies.

It remains to show that in every sequential equilibrium, each type of each player reports truthfully in stage $1$; then sequentially optimal play is pinned down in stages 2 and 3, by the argument above. Suppose for a contradiction that there is a sequential equilibrium in which some type $\th_i$ of some player $i$ reports $\th_i' = 1 - \th_i$. Upon seeing this report $\th_i'$, player $-i$ will challenge in stage $2$ because they infer that with probability at least $1/2$ player $i$ misreported and hence will not defend (this probability is $1/2$ if type $1 - \th_i$ reports truthfully and $1$ otherwise). Therefore, type $\th_i$'s payoff from reporting $\th_i' = 1- \th_i$ is at most $-3 + 2 = -1$. But type $\th_i$ can guarantee a payoff of $0$ by 
 instead reporting truthfully, not challenging, and then defending if challenged. This is a contradiction.

On the other hand, the mechanism admits additional pure Bayes--Nash equilibria. Here is one example. In stage $1$, each player misreports their type. In stage $2$, neither player challenges. In stage $3$, each player who was challenged in stage $2$ defends. Defending in stage $3$ after having misreported is not sequentially optimal (no matter the player's beliefs about their opponent's payoff type). But anticipating this defense, it is optimal for each player not to challenge in stage $2$. And anticipating no challenge, each player is indifferent over all reporting strategies in stage $1$. 

Consider the social choice function $f \colon \Th \to A$ defined by $f(\th) = x^{\th}$. Our equilibrium analysis shows that (a) every pure sequential equilibrium of the mechanism induces $f$, but (b) there exists a pure Bayes--Nash equilibrium that induces a social choice function different from $f$. In the terminology of implementation theory, this mechanism fully implements $f$ in pure sequential equilibrium, but it does not fully implement $f$ in pure Bayes--Nash equilibrium. In fact, no mechanism fully implements $f$ in pure Bayes--Nash equilibrium; this follows from \citet{Jackson1991} because $f$ violates the condition of Bayesian monotonicity (as we verify below in the proof of \cref{res:KPT_counterexample}). Thus, moving from Bayes--Nash to sequential equilibrium makes more social choice functions implementable. We show next, however, that these undesirable Bayes--Nash outcomes can be approximately recovered as sequential equilibria if the information structure is slightly perturbed. 

\subsection{Perturbations} To perturb this setting, we suppose that each player does not directly observe their payoff type but instead observes a private signal, which we call their \emph{type}. For each player $i$, define the space $T_i$ of types by
\[
    T_i = \Set{ (t_i^{\th_i'}, z_i) : \th_i' \in \Th_i~\text{and}~z_i \in \{ 0,1 \} }.
\]
Let $T = T_1 \times T_2$. Fix $\e \in (0,1)$. Suppose that the payoff-type profile $\th$ and the type profile $t$ are drawn from the following distribution over $\Th \times T$. The distribution over payoff-type profiles is uniform, as in the unperturbed setting. For each payoff-type profile $\th \in \Th$, the conditional distribution over type profiles $t \in T$ is given by\footnote{For each player $i$, we omit the realization $(t_i^{1 - \th_i}, 1)$ since its conditional probability is $0$.}
\[
\begin{array}{c|ccc}
\pi^{\e} ( t | \th) & (t_2^{\th_2},0) & (t_2^{\th_2},1) & (t_2^{1-\th_2},0)\\
\hline
(t_1^{\th_1},0)
    & 1-\e & 0 & 0\\
(t_1^{\th_1},1)
    & 0 & 0 & \e/2\\
(t_1^{1-\th_1},0)
    & 0 & \e/2 & 0
\end{array}.
\]

This conditional distribution can be interpreted as follows.  Given the payoff-type profile $\th = (\th_1, \th_2)$, the ``default'' is for each player $i$ to receive signal $(t_i^{\th_i}, 0)$. However, with probability $\e$, the default signal profile is changed as follows. A player $i \in \{1,2\}$ is selected uniformly. The first component of player $i$'s signal is flipped from $t_i^{\th_i}$ to $t_i^{1 - \th_i}$, and the second component of player $-i$'s signal is flipped from $z_{-i} = 0$ to $z_{-i} = 1$. Under this interpretation,  each player observes only the final, possibly modified, signal. Note that player $i$ is uncertain of their own payoff type after observing $z_i  = 0$, and that the players' signals are correlated conditional on the payoff-type profile. 

Under this perturbed information structure, we claim that the following strategy profile can be supported as a sequential equilibrium. If player $i$ receives signal $(t_i^{\th_i'}, z_i)$, then they play as follows. In stage $1$, they report $1 - \th_i'$ if $z_i = 0$ and $\th_i'$ if $z_i = 1$. In stage $2$, they do not challenge. In stage $3$, if they were challenged in stage $2$, they defend if and only if they followed their equilibrium strategy in stage $1$. Below, we construct stage-3 beliefs that make the stage-3 strategy sequentially optimal. 
In stage 2, every report is on path, so each player anticipates that their opponent will defend if challenged. Hence, it is optimal not to challenge. And anticipating no challenge, each player is indifferent over all reporting strategies in stage $1$. 

To construct consistent beliefs in stage 3 so that, after being challenged, player $i$ is certain that their true payoff type coincides with the report specified by their equilibrium strategy in stage $1$, we consider two cases. 
\begin{enumerate}
    \item If player $i$'s signal is $(t_i^{\th_i'}, 0)$ then after being challenged, player $i$ infers that their opponent received a signal with $z_{-i} = 1$; this belief can be justified by a sequence of trembles in which their opponent's probability of challenging after observing $z_{-i} = 0$ becomes arbitrarily small relative to the probability of challenging after observing $z_{-i} = 1$. Player $i$ concludes that their own signal was modified from the default, so their true payoff type is $1 - \th_i'$, which is the stage-$1$ report specified by their equilibrium strategy. 
    \item If  player $i$'s signal is $(t_i^{\th_i'}, 1)$, they are certain that their own payoff type is $\th_i'$, which is the stage-$1$ report specified by their equilibrium strategy. 
\end{enumerate}

Under this sequential equilibrium, if the true payoff-type profile is $\th$, then with probability $1-\e$, each player receives the ``default'' signal and hence the induced decision is $x^{1 - \th_1, 1 - \th_2}$, which is different from $f(\th)$. Taking $\e$ arbitrarily small, this sequential equilibrium recovers the undesirable BNE outcome with arbitrarily high probability. 

Our construction here is tailored to the details of the two-player, binary-type social choice environment and the particular three-stage mechanism. In our main analysis, we provide a general construction that applies to any social choice environment (even with correlated payoff types and interdependent values), any finite multi-stage mechanism with observed actions, and any social choice correspondence that violates Bayesian monotonicity. 

\section{Implementation setting} \label{sec:setting}

\paragraph{Environment}
Consider a Bayesian social choice environment  $(\Th, \pi, A, (u_i)_{i=1}^{n})$, where $\Th = \prod_{i=1}^{n} \Th_i$ is a finite space of payoff-type profiles, $\pi \in \D(\Th)$ is a common prior with full support, $A$ is a finite set of collective decisions, and $u_i \colon A \times \Th \to \R$ is player $i$'s utility function. We allow for interdependent values: player $i$'s utility can depend on others' payoff types. We also allow players' payoff types to be correlated under the prior $\pi$. 

\paragraph{Social choices} A \emph{social choice function} (SCF) is a function $f \colon \Th \to A$, which assigns to each profile $\th$ a decision $f(\th)$. 
A \emph{social choice correspondence (SCC)} is a nonempty-valued correspondence $F \colon \Th \twoheadrightarrow A$, which assigns to each profile $\th$ a nonempty set $F(\th)$ of decisions. 

For each player $i$ and payoff type $\th_{i} \in \Th_i$, define the induced interim preference relation $\succsim_{\th_i}$ on the space of social choice functions as follows: 
\[
    f \succsim_{\th_i} f' \iff 
    \sum_{\th_{-i} \in \Th_{-i}} u_i ( f(\th), \th) \pi_{-i} (\th_{-i} | \th_i) \geq \sum_{\th_{-i} \in \Th_{-i}} u_i ( f'(\th), \th)\pi_{-i} (\th_{-i} | \th_i). 
\]
Define the strict relation $\succ_{\th_i}$ analogously with a strict inequality in place of the weak inequality.


\paragraph{BNE implementation}
We largely follow the framework of \cite{Jackson1991}. A mechanism is a pair $(M,g)$ consisting of a set $M = \prod_{i=1}^{n} M_i$ of message profiles and an outcome function $g \colon M \to A$. Given a mechanism $(M,g)$, a strategy for player $i$ is a map $\s_i \colon \Th_i \to M_i$. We view a strategy profile $\s = (\s_i)_{i=1}^{n}$ as a map from $\Th$ to $M$ via $\s(\th) = (\s_i(\th_i))_{i=1}^{n}$.  A strategy profile $\s$ is a \emph{Bayes--Nash equilibrium (BNE)} if for each player $i$, type $\th_i$, and strategy $\s_i'$, we have $g \circ \s \succsim_{\th_i} g \circ ( \s_i', \s_{-i})$. Following \cite{Jackson1991}, we define a BNE to be a \emph{pure} strategy profile. 


\begin{defn}[BNE implementation] A mechanism $(M,g)$ \emph{BNE implements} an SCC $F$ if the following hold: 
\begin{enumerate}
    \item For each BNE $\s$ of $(M,g)$, the map $g \circ \s$ is a selection from $F$;
    \item For each selection $f$ from $F$, there exists a BNE $\s$ of $(M,g)$ such that $g \circ \s = f$. 
\end{enumerate}
An SCC $F$ is \emph{BNE-implementable} if there exists a mechanism $(M,g)$ that BNE implements $F$. 
\end{defn}

\cite{Jackson1991} introduces a necessary condition for BNE implementation, termed Bayesian monotonicity. To state the definition, we need a few preliminaries. A \emph{deception} for player $i$ is a function $\a_i \colon \Th_i \to \Th_i$.  We view a deception profile $\a = (\a_i)_{i=1}^{n}$ as a map from $\Th$ to $\Th$ via $\a( \th) = (\a_i (\th_i))_{i =1}^{n}$. Given an SCF $f$, define for any player $i$ and type $\bar{\th}_i$ the SCF $f_{i,\bar{\th}_i}$ by $f_{i,\bar{\th}_i} (\th_i, \th_{-i}) = f( \bar{\th}_i, \th_{-i})$.

\begin{defn}[Bayesian monotonicity] A social choice correspondence $F$ is \emph{Bayesian monotone} if for each selection $f$ from $F$, and each deception profile $\a$ such that $f \circ \a$ is not a selection from $F$, there exists some player $i$, payoff type $\bar{\th}_i \in \Th_i$, and social choice function $f'$ such that $f' \circ \a \succ_{\bar{\th}_i} f\circ \a$ and, for each $\th_i \in \Th_i$, we have $f \succsim_{\th_i} f'_{i, \a_i(\bar{\th}_i)}$.
\end{defn}

From the necessity proof of Theorem 1 in \citet{Jackson1991}, Bayesian monotonicity is necessary for BNE implementability. The interpretation is that, in order to ensure that it is not an equilibrium for the players to follow the deception profile $\a$, there must be some SCF $f'_{i, \a_i(\bar{\th}_i)}$ that can be used in place of $f$ to reward type $\bar{\th}_i$ for ``blowing the whistle'' on their opponents for following $\a$. The conditions require that type $\bar{\th}_i$ strictly prefers to blow the whistle when the players follow $\a$, but that no type of player $i$ strictly prefers to blow the whistle if there is no deception. 

\paragraph{Sequential implementation}

Following \cite{AghionEtAl2012}, we consider sequential implementation using finite multi-stage mechanisms with observed actions. Formally, such a mechanism is a triple $(\HH, \HH', g)$, where $\HH$ is the finite set of nonterminal histories,  $\HH'$ is the finite set of terminal histories, and $g \colon \HH' \to A$ is the outcome function. The set $\HH \cup \HH'$ of all histories contains the initial history $\varnothing$ and is closed under the predecessor-of relation. For each history $h \in \HH \cup \HH'$, let
\[
    M (h) = \Set{ m : (h,m) \in \HH \cup \HH'~\text{and}~h~\text{is the direct predecessor of}~(h,m)}.
\]
For each $h \in \HH'$, we have $M(h) = \varnothing$. For each $h \in \HH$, we have $M(h) = \prod_i M_i(h)$, where each set $M_i (h)$ is nonempty. At each history $h \in \HH$, the players who are active at $h$ move simultaneously, where player $i$ is active at $h$ if and only if $|M_i(h)|>1$. All players observe these moves before choosing how to move in the next stage. Write $M_i = \prod_{h \in \HH} M_i(h)$ and $M = \prod_i M_i$. Given a profile $m$ in $M$, denote by $g(m)$ the outcome at the terminal history induced by $m$. Similarly, for each history $h \in \HH$, write $g(m;h)$ for the outcome induced by the profile $m$ after starting at the history $h$. Note that $g(m;h)$ depends only on the moves prescribed by $m$ at history $h$
and its successors.

Given a multi-stage mechanism $(\HH, \HH', g)$,  a strategy for player $i$ is a map $\s_i \colon \Th_i \to M_i$. View a strategy profile $\s = (\s_i)_{i=1}^{n}$ as a map from $\Th$ to $M$ via $\s(\th) = (\s_i(\th_i))_{i=1}^{n}$. A strategy $\s_i$ can be equivalently viewed as a function on $\Th_i \times \HH$ satisfying $\s_i (\th_i, h) \in M_i (h)$ for all $\th_i \in \Th_i$ and $h \in \HH$. A belief system for player $i$ is a map $\mu_i \colon \Th_i \times \HH  \to \D( \Th_{-i} )$, which specifies player $i$'s belief about others' payoff types as a function of their own payoff type and the history. An assessment is a pair $(\s, \mu)$ consisting of a strategy profile $\s$ and a belief-system profile $\mu$. 

Following \cite{Maskin1999}, \cite{Jackson1991}, and \cite{AghionEtAl2012}, we restrict to pure equilibrium strategy profiles. In order to define belief consistency, however, we must consider behavioral strategies. A behavioral strategy $\tilde{\s}_i$ for player $i$ assigns to each pair $(\th_i, h) \in \Th_i \times \HH$ a move distribution $\tilde{\s}_i (\th_i, h) \in \D( M_i (h))$. A behavioral strategy is \emph{totally mixed} if $\supp \tilde{\s}_i (\th_i, h) = M_i (h)$ for all $\th_i \in \Th_i$ and $h \in \HH$. Given any pure strategy $\s_i$, we denote by $\d_{\s_i}$ the equivalent behavioral strategy satisfying $\d_{\s_i} (\th_i, h) = \d_{\s_i(\th_i, h)}$.

\begin{defn}[Sequential equilibrium] \label{def:SE}
An assessment $(\s, \mu)$ is a \emph{sequential equilibrium} if the following hold: 
\begin{enumerate}
    \item \emph{Sequential optimality}: For each player $i$, payoff type $\th_i \in \Th_i$, and history $h \in \HH$, 
    \[
        \s_i ( \th_i) \in \argmax_{m_i \in M_i} \sum_{\th_{-i} \in \Th_{-i}} u_i \bigl( g ( m_i, \s_{-i} ( \th_{-i}); h), \th \bigr) \mu_i ( \th_i, h) [ \th_{-i}].
    \]
    \item \emph{Consistency}: There exists a sequence, $(\s^k, \mu^k)$,  of totally mixed behavioral strategy profiles and belief-system profiles such that for each player $i$:
    \begin{enumerate}
        \item for each $k$, the belief system $\mu_i^k$ is derived from $\s_{-i}^k$ via Bayes' rule;
        \item the sequence $(\s_i^k, \mu_i^k)$ converges to $(\d_{\s_i}, \mu_i)$ pointwise on $\Th_i \times \HH$.
    \end{enumerate}
\end{enumerate}
\end{defn}

\begin{defn}[Sequential implementation]
A multi-stage mechanism $(\HH, \HH', g)$ \emph{sequentially implements} a social choice correspondence $F$ if the following hold.
\begin{enumerate}
    \item For each sequential equilibrium $(\s, \mu)$ of $(\HH, \HH', g)$, the map $g \circ \s$ is a selection from $F$. 
    \item For each selection $f$ from $F$, there exists a sequential equilibrium $(\s, \mu)$ of $(\HH, \HH', g)$ such that $g \circ \s = f$. 
\end{enumerate}
An SCC $F$ is \emph{sequentially implementable} if there exists a multi-stage mechanism $(\HH, \HH',g)$ that sequentially implements $F$. 
\end{defn}

\section{Information perturbations} \label{sec:perturbations}

We now define information perturbations of the social choice environment $(\Th, \pi, A, (u_i)_{i=1}^{n})$. Rather than perfectly observing their true payoff type, each player observes a private signal about it, which determines their type. Following \cite{FKL1988}, we associate each payoff type of each player with a distinguished type that will play the role of the ``default'' signal in the example from \cref{sec:example}. Formally, an \emph{information structure} is a triple $(T, t^\ast, \hat{\pi})$ specifying a finite set $T = \prod_{i} T_i$ of type profiles, an injective map $t_i^\ast \colon \Th_i \to T_i$ for each player $i$, and a prior $\hat{\pi} \in \D( \Th \times T)$ such that for each player $i$, the marginal $\marg_{T_i} \hat{\pi}$ has full support on $T_i$.  This definition implies that $|T_i| \geq |\Th_i|$ for each player $i$. 

A player's type determines their beliefs about the profile of payoff types, but the types do not directly affect payoffs. We write $t^\ast (\th) = (t_i^\ast(\th_i))_{i=1}^{n}$ for the distinguished type profile associated to the payoff-type profile $\th$. Define the distribution $\pi \otimes t^\ast \in \D( \Th \times T)$ by  $(\pi \otimes t^\ast)[ \th, t^\ast (\th)] = \pi (\th)$. Thus, the original setting can be represented by  any information structure $(T, t^\ast, \hat{\pi})$ satisfying $\hat{\pi} = \pi \otimes t^\ast$.\footnote{In this case, the full marginal support condition implies that  $T_i = t_i^\ast (\Th_i)$ for all $i$.}

Given an information structure $(T, t^\ast, \hat{\pi})$ and a multi-stage mechanism $(\HH, \HH', g)$, a strategy for player $i$ is a map $\s_i \colon T_i \to M_i$. A strategy $\s_i$ can equivalently be viewed as a function on $T_i \times \HH$ satisfying $\s_i (t_i, h) \in M_i (h)$ for all $t_i \in T_i$ and $h \in \HH$. A belief system for player $i$ is a map $\mu_i \colon T_i \times \HH  \to \D( \Th \times T_{-i})$, which specifies player $i$'s belief about the payoff-type profile and others' types as a function of their own type and the history.  An assessment $(\s, \mu)$ is a sequential equilibrium if it is sequentially optimal and consistent, where these properties are defined as in \cref{def:SE}, mutatis mutandis. 

We will consider parametric families of information structures, denoted by $(T, t^\ast, \pi^{\e})_{\e \in (0,1)}$. Note that $T$ and $t^\ast$ do not vary with $\e$. Let $\| \cdot \|_{\mathrm{TV}}$ denote the total variation norm. That is, for any $\pi_1, \pi_2 \in \D( \Th \times T)$, let $\| \pi_1 - \pi_2 \|_{\mathrm{TV}} = \sup_{E} | \pi_1 (E) - \pi_2(E) |$, where the supremum is over all subsets $E$ of $\Th \times T$.

\begin{defn}[Convergence of information structures] \label{def:conv} A family of information structures, $(T, t^\ast, \pi^{\e})_{\e \in (0,1)}$,  converges to $\pi$ if
\begin{enumerate}[label = (\roman*)]
    \item \label{it:marg} $\marg_{\Th} \pi^{\e} = \pi$ for all $\e$;
    \item \label{it:limit_TV} $\lim_{\e \to 0} \| \pi^{\e} - \pi \otimes t^\ast \|_{\mathrm{TV}} = 0$.
\end{enumerate}
\end{defn}

Part~\ref{it:limit_TV} implies that $\lim_{\e \to 0} \| \marg_{\Th} \pi^\e - \pi \|_{\mathrm{TV}} = 0$. We impose the stronger requirement \ref{it:marg} because it strengthens our main non-robustness result, \cref{res:fragile}. 

To interpret this notion of convergence, consider the event $E^\ast$ that each player's type is the distinguished type for their true payoff type. Formally, $E^\ast = \{ ( \th, t^\ast (\th)): \th \in \Th \}$.  Note that $(\pi \otimes t^\ast ) (E^\ast)  = 1$. Suppose that a family $(T, t^\ast, \pi^{\e})_{\e \in (0,1)}$ converges to $\pi$. Then there exists a function $p \colon (0,1) \to [0,1]$ satisfying $\lim_{\e \to 0} p(\e) = 1$ such that for each $\e \in (0,1)$, the event $E^\ast$ has ex-ante probability at least $p(\e)$ and $E^*$ is an evident $p(\e)$-belief in the sense of \citet[p.~176]{MondererSamet1989}.\footnote{Indeed, let $p(\e) = \min \bigl\{ \pi^{\e} (E^\ast ), \min_{i, \th_i} \pi^{\e} [ E^\ast | t_i^\ast (\th_i)] \bigr\}$.}

\begin{rem}[Fixed payoffs versus payoff perturbations]
In the perturbations  we consider, the utility functions $u_i$ are fixed, and signals affect behavior only through players' beliefs. A broader class of
perturbations would allow the utility functions to depend on players' types and require only that the utility functions associated with the distinguished types converge to the baseline utility functions. As in the complete-information case studied by \cite{FKL1988},  using this broader class of perturbations can enlarge the set of baseline equilibrium behaviors that can arise as
limits of sequential  equilibria under perturbations: revealing types may provide direct incentives for
an action through their different payoffs, rather than only changing the beliefs
induced by their presence. Establishing fragility in the more restrictive
fixed-payoff case strengthens our result.
\end{rem}

\section{Fragility of sequential implementation} \label{sec:fragility}

We now show that the expanded implementation power from sequential equilibrium relative to Bayes--Nash equilibrium is not robust to small information perturbations. 

\begin{thm}[Necessity of Bayesian monotonicity]
 \label{res:fragile} Let $F$ be a social choice correspondence that is not Bayesian monotone. If $F$ is sequentially implemented (under $\pi$) by some multi-stage mechanism $(\HH, \HH', g)$, then there exists a family of information structures, $(T, t^\ast, \pi^\e)_{\e \in (0,1)}$, converging to $\pi$ and an associated family $(\s^{\e}, \mu^{\e})_{\e \in (0,1)}$ of sequential equilibria of $(\HH, \HH', g)$ such that for some $\th \in \Th$ and all $\e \in (0,1)$, the decision $g(\s^{\e} (t^\ast (\th)))$ is not in $F(\th)$. 
\end{thm}

\begin{rem}[Integer games] \label{res:integer} We require that $\HH \cup \HH'$ is finite, while the mechanisms in \cite{BerginSen1998} and \cite{Baliga1999} ask players to select arbitrary integers at some histories. In the proof, we show that our conclusion still holds for any finite-horizon multi-stage mechanism $(\HH, \HH', g)$ with countably many moves, provided that it admits a sequential best response for each player at each history, for all beliefs and opposing pure strategies. This property holds for the mechanisms in \cite{BerginSen1998} and \cite{Baliga1999}. 
\end{rem}

To interpret the conclusion of \cref{res:fragile}, recall from the definition of convergence that for each $\th \in \Th$, we have $\lim_{\e \to 0} \pi^{\e} [ \th, t^\ast(\th)] = \pi (\th)$. Thus, the limiting ex-ante probability that the induced decision is outside $F$ is at least $\min_{\th' \in \Th} \pi (\th')$, which is strictly positive.

Here is an outline of the proof. Let $F$ be an SCC that violates Bayesian monotonicity, and suppose that the multi-stage mechanism $(\HH,\HH',g)$ sequentially implements $F$ under $\pi$. Then there exists a sequential equilibrium $(\hat{\s},\hat{\mu})$ of $(\HH,\HH',g)$ such that $g\circ\hat{\s}$ is a selection from $F$, and there exists a deception profile $\a$ such that $\hat{\s}\circ\a$ is a Bayes--Nash equilibrium of $(\HH,\HH',g)$ but $g\circ\hat{\s}\circ\a$ is not a selection from $F$.

For each $\e \in (0,1)$, define an information structure and an associated strategy profile such that, for every $\th \in \Th$, conditional on the payoff-type profile being $\th$, the players coordinate on the profile $\hat{\s}(\a(\th))$ with probability $1-\e$ and on the profile $\hat{\s}(\th)$ with probability $\e$. Each player's signal provides them with enough information to coordinate play in this way, but at least one player is sometimes uncertain of their payoff type and hence about which of the two profiles the players are coordinating on. We define the belief system so that if any player observes a deviation by an opponent, they infer that the players are coordinating on $\hat{\s}(\th)$ (though they may remain uncertain of the true payoff-type profile $\th$). We then derive sequential optimality and consistency from the corresponding properties of $(\hat{\s}, \hat{\mu})$ under $\pi$. 

\cref{res:fragile} is an incomplete-information analogue of the complete-information result of \citet[Theorem 3]{AghionEtAl2012}. The two  constructions have a common structure: The failure of Bayesian/Maskin monotonicity is used to construct a ``bad'' Bayes--Nash/Nash equilibrium. Then an information structure is defined in order to induce a mixture between the good and bad equilibria. To ensure sequential optimality, each player infers from a detectable opponent deviation that the players are coordinating on the good equilibrium. Despite these similarities,  the incomplete-information setting requires a different and more intricate construction.

First, the state realization can no longer serve as a coordination device. In order to coordinate play on the good and bad profiles, we enrich each player's signal space so that it is strictly larger than their payoff-type space. Each player's signal indicates how they should play without necessarily revealing their own payoff type. Second, even after a player observes a deviation by an opponent, they may remain uncertain about their opponents' payoff types. We construct a sequence of trembles that replicates the potentially nondegenerate off-path beliefs about opponents' payoff types that guarantee sequential optimality under the ``good'' equilibrium. Third, our construction is global in the sense that it depends on the full deception profile $\a$. By contrast, the failure of Maskin monotonicity yields a single ``deception'' for all players, which differs from the identity at a single state. Thus, \cite{AghionEtAl2012} reason locally about a fixed pair of states $\th'$ and $\th''$; for every other realization, the state becomes common knowledge.

With additional structure on the social choice environment, Bayesian monotonicity and incentive compatibility are jointly sufficient for BNE implementation. Given social choice functions $f$ and $f'$ and a subset $E$ of $\Th$, let $f_Ef'$ denote the social choice function that coincides with $f$ on $E$ and with $f'$ on $\Th \setminus E$. 

\begin{defn}[Economic environment] \label{def:ec} The social choice environment $(\Th, \pi, A, (u_i)_{i=1}^{n})$ is \emph{economic} if for each SCF $f$ and each payoff-type profile $\th \in \Th$, there exist distinct players $i$ and $j$ and constant SCFs $a$ and $b$ such that for all subsets $E$ of $\Th$ containing $\th$, we have $a_E f \succ_{\th_i} f$ and $b_E f \succ_{\th_j} f$. 
\end{defn}

\cref{def:ec} is a uniform non-satiation property, which usually holds in exchange economies. A social choice correspondence is incentive compatible if every selection from it is Bayesian incentive compatible in the standard sense. By \citet[Theorem 1, p.~467]{Jackson1991}, if $n \geq 3$ and the environment is economic, then Bayesian monotonicity and incentive compatibility are jointly sufficient for BNE implementability.\footnote{\cite{Jackson1991} requires an additional condition, termed closure, but this condition is already implied by our focus on social choice correspondences, and the associated set of selections, as opposed to general sets of social choice functions.} Incentive compatibility is also necessary for sequential implementability (and BNE implementability); otherwise, any selection that is not incentive compatible could never be induced by any BNE since some type would strictly prefer to mimic the strategy of another type. Thus, we obtain the following. 

\begin{cor}[Fragility of additional implementation]
Suppose that $n \geq 3$ and the social choice environment is economic. Let $F$ be a social choice correspondence that is not BNE implementable (under $\pi$). If $F$ is sequentially implemented (under $\pi$) by some multi-stage mechanism $(\HH, \HH', g)$, then the conclusion of \cref{res:fragile} holds.
\end{cor}

\begin{rem}[Bayes--Nash implementation]
To highlight the contrast between sequential implementation and Bayes--Nash 
implementation, note that a standard upper hemicontinuity property of BNE gives the following result. Suppose that a social choice correspondence $F$ is implemented in BNE (under $\pi$) by some finite mechanism $(M,g)$. For any family of information structures, $(T, t^\ast, \pi^\e)_{\e \in (0,1)}$, converging to $\pi$ and any associated family $(\s^{\e})_{\e \in (0,1)}$ of pure BNE, there exists $\bar{\e} > 0$ such that $g(\s^{\e} (t^\ast (\th))) \in F(\th)$ for all $\e \in (0, \bar{\e})$ and all $\th \in \Th$.\footnote{Otherwise, since $T$ and $M$ are finite, there exists some payoff-type profile $\th \in \Th$, strategy profile $\s \colon T \to M$, and sequence $(\e^k)$ converging to $0$ such that for all $k$,  we have $\s^{\e^k} = \s$ and $g( \s( t^\ast (\th))) \notin F(\th)$.  Passing to the limit in the equilibrium conditions, it follows that the profile $\s'(\th) = \s ( t^\ast (\th))$ is a BNE of $(M,g)$ under $\pi$, contrary to BNE implementation.}
\end{rem}

\section{More restrictive perturbations} \label{sec:KPT}

To prove \cref{res:fragile}, we constructed information structures with two properties: (a) at least one player is sometimes uncertain of their own payoff type; and (b) conditional on the payoff-type profile, the players' types are not statistically independent. Here, we show that the general result fails if perturbations are required to satisfy ``known payoff types'' or conditional independence.

\subsection{Known payoff types}

An information structure $(T, t^\ast, \hat{\pi})$ has \emph{known payoff types (KPT)} if (a) for each player $i$, there exists a finite set $S_i$ such that $T_i \subseteq \Th_i \times S_i$, and (b) $\hat{\pi} (\th, (\th_i', s_i)_{i=1}^{n}) = 0$ for every $\th \in \Th$ and $(\th_i', s_i)_{i=1}^{n} \in T$ such that $\th \neq (\th_i')_{i=1}^{n}$. Intuitively, each player $i$'s type $t_i = (\th_i, s_i)$ reveals their true payoff type $\th_i$ and a supplementary signal $s_i \in S_i$.

 \begin{rem}[KPT and knowledge of payoffs]
The KPT condition should not be interpreted as requiring player $i$ to know their
entire realized payoff function. KPT requires that player $i$ observe their own
payoff-type component $\theta_i$. Because we allow interdependent values, player
$i$ may nevertheless remain uncertain about the utility consequences of a decision:
their utility $u_i(a,\theta)$ can depend on $\theta_{-i}$. A stronger known-payoffs property would require that all payoff-type profiles that player $i$ considers possible induce the same function $a\mapsto u_i(a,\theta)$.
\end{rem}

\begin{prop}[KPT counterexample] \label{res:KPT_counterexample} There is a social choice environment $(\Th, \pi, A, (u_i)_{i=1}^{n})$ with independent private values, an SCF $f$ that is not Bayesian monotone, and a mechanism $(\HH, \HH', g)$ that sequentially implements $f$ under $\pi$ such that the following holds: There exists $\bar{\h} > 0$ such that for every KPT information structure $(T, t^\ast, \hat{\pi})$ satisfying $\| \hat{\pi} - \pi \otimes t^\ast \|_{\mathrm{TV}} < \bar{\h}$,  every sequential equilibrium $(\s, \mu)$ of $(\HH, \HH', g)$ satisfies $g ( \s ( t^\ast(\th))) = f(\th)$ for all $\th \in \Th$.
\end{prop}

The final conclusion of \cref{res:KPT_counterexample} implies that
\[
    \hat{\pi} \Set{ (\th, t) \in \Th \times T: g(\s(t)) \neq f(\th) } \leq \| \hat{\pi} - \pi \otimes t^\ast \|_{\mathrm{TV}},
\]
so in the example, the error probability on the left side converges to $0$ for any family of KPT information structures converging to $\pi$. 

To prove \cref{res:KPT_counterexample},  we return to the example from \cref{sec:example}. Let $(\HH, \HH', g)$ denote the mechanism and $f$ the associated SCF from \cref{sec:example}. In the perturbations from that example, it is critical that each player is sometimes uncertain of their own payoff type because then, after observing an unexpected challenge by their opponent, the player infers that the opponent's signal reveals that the player's payoff type matches the report they submitted. If players know their own payoff types, this kind of inference is impossible. 


\subsection{Conditional independence}

An information structure $(T, t^\ast, \hat{\pi})$ is \emph{conditionally independent (CI)} if under $\hat{\pi}$, the types $t_1, \ldots, t_n$ are mutually independent, conditional on $\th$. This independence property implies that $\hat{\pi} ( t_{-i} | \th, t_i) = \hat{\pi} (t_{-i} | \th)$ whenever $\hat{\pi} ( \th, t_i) > 0$. This is a natural generalization of \citeapos{FKL1988} independent-types condition to our setting. Next, we prove the analogue of \cref{res:KPT_counterexample} for CI perturbations in place of KPT perturbations.

\begin{prop}[CI counterexample] \label{res:CI_counterexample}  There is a social choice environment $(\Th, \pi, A, (u_i)_{i=1}^{n})$ with independent private values, an SCF $f$ that is not Bayesian monotone, and a mechanism $(\HH, \HH', g)$ that sequentially implements $f$ under $\pi$ such that the following holds:  There exists $\bar{\h} > 0$ such that for every CI information structure $(T, t^\ast, \hat{\pi})$ satisfying $\| \hat{\pi} - \pi \otimes t^\ast \|_{\mathrm{TV}} < \bar{\h}$,  every sequential equilibrium $(\s, \mu)$ of $(\HH, \HH', g)$ satisfies $g ( \s ( t^\ast(\th))) = f(\th)$ for all $\th \in \Th$.
\end{prop}

To prove \cref{res:CI_counterexample}, we  consider a variant of the example from \cref{sec:example}. The mechanism still has the report--challenge--defend structure, but now only player $2$ can challenge, and only after a particular report by player $1$. If there is no challenge, the payoffs now depend on the reports, unlike the example from \cref{sec:example}. Under CI information structures, players can be uncertain of their own payoff types. Thus, an unexpected challenge can cause a player to update their beliefs about their own payoff type. Conditional independence implies that a player's signal does not change the conditional distribution of their opponent's signal given the payoff-type profile. In the example, the ``bad'' BNE outcome cannot be sustained as a sequential equilibrium because the payoff-type beliefs that can justify defending after a challenge in turn constrain beliefs about the opponent's signal in a way that makes a stage-1 reporting deviation profitable.

\section{Discussion} \label{sec:discussion}
We show that the additional implementation power of sequential
equilibrium is fragile: arbitrarily small information perturbations
can produce sequential equilibria with inadmissible outcomes. Our counterexamples show that this
conclusion need not hold when perturbations preserve each player's
knowledge of their own payoff type or require players' signals to be
conditionally independent given the payoff-type profile.

Our results leave open whether implementation under other refinements such as undominated
Nash equilibrium, strategic stability \citep{kohlberg1986strategic}, or
strict Nash equilibrium, is robust to  information perturbations.  If so, it would be interesting to see if they are also robust to slightly expanded notions of nearby games, such as replacing  the condition $\marg_{\Th} \pi^{\e} = \pi$ 
in Definition \ref{def:conv} with its limiting version  $\lim_{\e \to 0}\marg_{\Th} \pi^{\e} = \pi$ as suggested by \citet{FrickRomm2015}.

\appendix

\section{Proofs} \label{sec:proof_lemmas}

\subsection{Proof of \cref{res:fragile}} \label{sec:main_proof}

Fix a social choice correspondence $F$ that is not Bayesian monotone and a multi-stage mechanism $(\HH, \HH', g)$ that sequentially implements $F$ (under $\pi$). 

The proof proceeds in three steps. First, we obtain a suitable profile of deceptions. Second, we use this deception profile to define the family of information structures. Third, we construct a sequential equilibrium under each information structure with the desired properties. 

\paragraph{Obtaining the deception profile}
 Since $F$ is not Bayesian monotone, there exists a selection $f$ from $F$ and a profile $\a = (\a_i)_{i=1}^{n}$ of deceptions such that (i) $f \circ \a$ is not a selection from $F$ and  (ii) for each player $i$ and type $\bar{\th}_i$ and each $f' \colon \Th \to A$, if $f \succsim_{\th_i} f'_{i, \a_i(\bar{\th}_i)}$ for all types $\th_i$ in $\Th_i$, then $f \circ \a \succsim_{\bar{\th}_i} f' \circ \a$. Note that $\a_i$ is not the identity for at least one player $i$. Since $f$ is a selection from $F$ and $(\HH, \HH', g)$ sequentially implements $F$, there exists a sequential equilibrium $(\hat{\s}, \hat{\mu})$ such that $g \circ \hat{\s} = f$. 

\begin{lem}[Deceptive BNE] \label{res:BNE_deception} Under $\pi$, the profile $\hat{\s} \circ \a = (\hat{\s}_i \circ \a_i)_{i=1}^{n}$ is a Bayes--Nash equilibrium of $(\HH,\HH', g)$. 
\end{lem}

\begin{proof} We follow the proof in \cite{Jackson1991}. Fix player $i$, type $\th_i \in \Th_i$, and a deviation $m_i \in M_i$. It suffices to show that type $\th_i$ cannot profit by deviating to $m_i$. Define $\s_i' \colon \Th_i \to M_i$ by $\s_i'( \th_i') = m_i$ for all $\th_i'$. Define $f' \colon \Th \to A$ by $f' = g \circ (\s_i', \hat{\s}_{-i})$. By construction, $f = g \circ \hat{\s}$. Since $\hat{\s}$ is a BNE, we conclude that
\[
    f = g \circ \hat{\s} \succsim_{\th_i'} g \circ (\s_i', \hat{\s}_{-i}) = f' = f'_{i, \a_i (\th_i)},
\]
for all types $\th_i'$. Hence, 
\[
    g \circ (\hat{\s} \circ \a) = f \circ \a \succsim_{\th_i} f' \circ \a = g \circ ( \s_i', \hat{\s}_{-i} \circ \a_{-i}). \qedhere
\]
\end{proof}

\paragraph{Constructing the information structures}

For each player $i$, let 
\begin{equation} \label{eq:S_i_def}
    T_i = \Set{ (t_i^{\th_i}, 0): \th_i \in \Th_i } \cup  \Set{ (t_i^{\th_i}, 1): \th_i \in \Th_i } \cup \Set{ (t_i^{\th_i}, 2): \th_i \in \Th_i \setminus \a_i (\Th_i) }.
\end{equation}
Let $T = \prod_{i=1}^{n} T_i$.  Define $t_i^\ast \colon \Th_i \to T_i$ by $t_i^\ast(\th_i) = (t_i^{\th_i}, 0)$. 

For each player $i$, define $\bar{\a}_i \colon T_i \to \Th_i$ by 
\begin{equation} \label{eq:abar}
    \bar{\a}_i ( t_{i}^{\th_i'}, z_i) = 
    \begin{cases}
        \a_i (\th_i') &\text{if}~z_i =0, \\
        \th_i' &\text{if}~z_i \neq 0.
    \end{cases}
\end{equation}
We construct a joint distribution over $\Th \times T$ that depends on the realization of a latent Bernoulli variable $Z$ with the following properties. If $Z = 0$, then $\bar{\a}_i (t_i) = \a_i(\th_i)$. If $Z = 1$, then $\bar{\a}_i(t_i) = \th_i$. And every type profile is consistent with $Z = 1$. 

Now we turn to the formal construction. First, we introduce auxiliary kernels. For each $z \in \{0,1\}$ and each player $i$, define $d_i^z \colon \Th_i \to \D(T_i)$ by letting $d_i^z (\th_i)$ equal the point mass on $(t_i^{\th_i}, z)$. Define the product kernel $d^z \colon \Th \to \D(T)$ by $d^z(\th) = \otimes_{i} d_i^z (\th_i)$. Next, define $\k_i \colon \Th_i \to \D(T_i)$ as follows. For $\th_i \in \a_i (\Th_i)$, let
\[
    \k_i( t_i^{\th_i'}, z_i | \th_i) 
    =
     \begin{cases} 
     1/2 &\text{if}~(\th_i', z_i) = (\th_i, 1), \\
     \frac{1}{2 | \a_i^{-1} (\th_i)|} &\text{if}~ \th_i' \in \a_i^{-1} (\th_i)~\text{and}~ z_i = 0, \\
     0 &\text{otherwise}. 
\end{cases}
\]
For $\th_i \in \Th_i \setminus \a_i (\Th_i)$, let
\[
    \k_i( t_i^{\th_i'}, z_i | \th_i) 
    =
     \begin{cases} 
     1/2 &\text{if}~(\th_i', z_i) = (\th_i, 1), \\
     1/2 &\text{if}~ (\th_i', z_i) = (\th_i, 2), \\
     0 &\text{otherwise}. 
\end{cases}
\]
Define the product kernel $\k \colon \Th \to \D(T)$ by $\k (\th) = \otimes_i \k_i (\th_i)$. 

For each $\e \in (0,1)$, define $\pi^{\e} \in \D (\Th \times T)$ by
$
    \pi^{\e} = (1 - \e) (\pi \otimes d^0) + \e (\pi \otimes \k), $
where $\otimes$ denotes the product of a measure and a kernel.\footnote{That is, $(\pi \otimes \k)[\th, t] = \pi (\th) \k (t | \th)$ for all $\th \in \Th$ and $t \in T$.} We include $z_i=2$ in order to keep the probability of $z_i = 1$ constant across payoff types. This is not essential, but it simplifies the construction of trembles below.

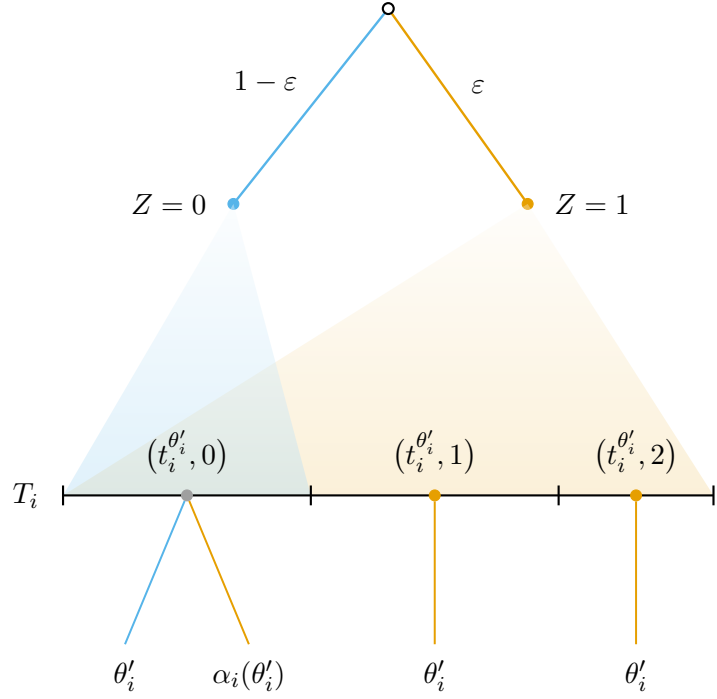
\begin{figure}
\centering
\begin{tikzpicture}[
    x=0.82cm,
    y=0.82cm,
    every node/.style={font=\small},
    neutral/.style={black, line width=0.8pt},
    blue branch/.style={Blue, line width=0.9pt},
    orange branch/.style={Orange, line width=0.9pt},
    relation/.style={line width=0.75pt},
    blue relation/.style={relation, Blue},
    orange relation/.style={relation, Orange},
    root point/.style={circle, draw=black, fill=white, line width=0.8pt, inner sep=1.45pt},
    blue point/.style={circle, draw=Blue, fill=Blue, inner sep=1.45pt},
    orange point/.style={circle, draw=Orange, fill=Orange, inner sep=1.45pt},
    ambiguous point/.style={circle, draw=gray!75, fill=gray!75, inner sep=1.45pt}
]

\coordinate (root) at (5.25,7.85);
\coordinate (Zzero) at (2.75,4.70);
\coordinate (Zone) at (7.50,4.70);

\draw[blue branch] (root) -- node[pos=.50, above left=1pt, text=black] {$1-\varepsilon$} (Zzero);
\draw[orange branch] (root) -- node[pos=.50, above right=1pt, text=black] {$\varepsilon$} (Zone);
\node[root point] at (root) {};
\node[blue point] at (Zzero) {};
\node[orange point] at (Zone) {};
\node[left=6pt] at (Zzero) {$Z=0$};
\node[right=6pt] at (Zone) {$Z=1$};

\shade[top color=Orange!5, bottom color=Orange!45, opacity=.33]
    (Zone) -- (0,0) -- (10.5,0) -- cycle;

\shade[top color=Blue!5, bottom color=Blue!45, opacity=.40]
    (Zzero) -- (0,0) -- (4,0) -- cycle;

\draw[neutral] (0,0) -- (10.5,0);
\foreach \x in {0,4,8,10.5}
    \draw[neutral] (\x,-0.16) -- (\x,0.16);

\node[left=5pt] at (0,0) {$T_i$};
\node[ambiguous point] (s0) at (2,0) {};
\node[orange point] (s1) at (6,0) {};
\node[orange point] (s2) at (9.25,0) {};

\node[above=4pt] at (s0) {$\bigl(t_i^{\theta_i'},0\bigr)$};
\node[above=4pt] at (s1) {$\bigl(t_i^{\theta_i'},1\bigr)$};
\node[above=4pt] at (s2) {$\bigl(t_i^{\theta_i'},2\bigr)$};

\coordinate (theta0) at (1.00,-2.40);
\coordinate (alpha0) at (3.00,-2.40);
\coordinate (theta1) at (6,-2.40);
\coordinate (theta2) at (9.25,-2.40);

\draw[blue relation] (s0) -- (theta0);
\draw[orange relation] (s0) -- (alpha0);
\draw[orange relation] (s1) -- (theta1);
\draw[orange relation] (s2) -- (theta2);

\node[below=2pt] at (theta0) {$\theta_i'$};
\node[below=2pt] at (alpha0) {$\alpha_i(\theta_i')$};
\node[below=2pt] at (theta1) {$\theta_i'$};
\node[below=2pt] at (theta2) {$\theta_i'$};

\end{tikzpicture}

\caption{Visualization of the distribution of $(Z,t_i, \th_i)$ under $\pi^{\e}$}
\label{fig:info_structure}
\end{figure}

\begin{rem}[Interpretation of information structure] 
\cref{fig:info_structure} shows how the latent variable $Z$ generates
signals and, together with each signal, determines the payoff type.
Types with $z_i=0$ are \emph{ambiguous}: they can arise under either
value of $Z$. Types with $z_i\neq 0$ show 
that $Z=1$ so they are \emph{revealing}. Conditional on $Z=1$, every type profile has positive
probability, a property we use to construct beliefs after deviations.
\end{rem}

The family $(T, t^\ast, \pi^{\e})_{\e \in (0,1)}$ converges to the original setting $\pi$. Clearly, $\marg_{\Th} \pi^{\e} = \pi$ for each $\e$. And by the definition of $t^\ast$, we have $\pi \otimes d^0 = \pi \otimes t^\ast$, so $\| \pi^{\e} - \pi \otimes t^\ast \|_{\mathrm{TV}} \leq \e$.

\paragraph{Constructing the sequential equilibrium}

Recall the definition of $\bar{\a}$ from \eqref{eq:abar}. Conditional on $Z = 0$, we have $\hat{\s}(\bar{\a}(t)) = \hat{\s} (\a ( \th))$. Conditional on $Z = 1$, we have $\hat{\s}(\bar{\a}(t)) = \hat{\s} (\th)$. 

\begin{lem}[BNE under perturbation] \label{res:BNE_elaboration} Under $\pi^\e$, the profile $\hat{\s} \circ \bar{\a} = (\hat{\s}_i \circ \bar{\a}_i)_{i=1}^{n}$  is a Bayes--Nash equilibrium of $(\HH, \HH', g)$. 
\end{lem}

\begin{proof}
Fix player $i$ and a type $t_i = (t_i^{\th_i'}, z_i) \in T_i$. To show that player $i$ does not have a profitable deviation, it suffices to show that player $i$ does not have a profitable deviation conditional on each realization of $Z$.\footnote{Here and below, when computing probabilities involving $Z$, we think of $\pi^{\e}$ as being extended to $\D( \Th \times T \times \{0,1\})$ as $(1 - \e) (\pi \otimes d^0) \otimes \d_0 + \e (\pi \otimes \k) \otimes \d_1$. The symbol $\otimes$ is overloaded: it can denote the product between two measures or between a measure and a kernel.} We use the notation $[\cdot]$ to denote the indicator function for the predicate it encloses.

\begin{itemize}
    \item We have $\pi^{\e} [Z = 1| t_i] > 0$ and 
\[
     \pi^{\e} ( \th | t_i, Z = 1) = \pi (\th_{-i}| \bar{\a}_i(t_i)) [ \th_i = \bar{\a}_i (t_i)].
 \]
In this case, player $i$ is certain that the players are following $\hat{\s} (\th)$, so they have no profitable deviation since $\hat{\s}$ is a BNE (under $\pi$).

\item We have $\pi^{\e} [Z = 0| t_i] > 0$ if and only if $z_i = 0$. If $z_i = 0$, then 
\[
    \pi^{\e} ( \th | t_i, Z = 0) = \pi (\th_{-i}| \th_i')[\th_i = \th_i'].
\]
In this case, player $i$ is certain that the players are following $\hat{\s} (\a(\th))$, so they have no profitable deviation since $\hat{\s} \circ \a$ is a BNE (under $\pi$) by \cref{res:BNE_deception}. \qedhere
\end{itemize}
\end{proof}

To turn $\hat{\s}\circ\bar{\a}$ into a sequential equilibrium $(\s, \mu) = (\s^{\e}, \mu^{\e})$ without changing its induced outcomes, we define beliefs and modify players'
strategies at certain off-path histories.  To begin, we introduce notation to classify histories. For each player $i$ and history $h \in \HH$, let $C_i(h)$ be the set of payoff types $\th_i$ in $\Th_i$ such that the sequence of moves taken by player $i$ to reach $h$ is consistent with $\hat{\s}_i (\th_i)$. Let $\dev (h) = \{ j : C_j (h) = \varnothing\}$. In words, $\dev (h)$ is the set of players who, at history $h$, are revealed to have deviated. For each $i$, let $\dev_i (h) = \dev (h) \setminus \{i\}$.

Fix player $i$. Define $\mu_i \colon T_i \times \HH \to \D( \Th \times T_{-i})$ as follows. 
\begin{itemize}
    \item If $\dev_i (h) = \varnothing$, let
\begin{equation} \label{eq:Bayes}
     \mu_i (t_i, h) [\th, t_{-i}] =
     \frac{ \pi^\e ( \th, t_{-i} | t_i) \prod_{j \neq i} [ \bar{\alpha}_j (t_j) \in C_j (h)]}{\sum_{(\th', t_{-i}') \in \Th \times T_{-i}} \pi^\e ( \th', t_{-i}'| t_i) \prod_{j \neq i} [ \bar{\alpha}_j (t_j') \in C_j (h)]},
\end{equation}
where the denominator is positive because $\dev_i(h) = \varnothing$ and under $\pi^{\e} (\cdot |t_i)$, the random variable $\bar{\a}_{-i} (t_{-i})$ has full support on $\Th_{-i}$.


    \item If $\dev_i (h) \neq \varnothing$, let 
    \begin{equation} \label{eq:mu_def}
\mu_i (t_i, h) = \d_{\bar{\a}_i (t_i)} \otimes \Brac{ \hat{\mu}_{i} ( \bar{\a}_i(t_i), h) \otimes \k_{i}^h},
\end{equation}
where $\k_{i}^h \colon \Th_{-i} \to \D(T_{-i})$ is the product kernel $\otimes_{j \neq i} \k_{i,j}^h$, where

\[ 
    \k_{i,j}^h    = \begin{cases}
        d_j^1 &\text{if}~j \in \dev_i(h), \\
        \k_j &\text{if}~j \notin \dev_i(h). 
    \end{cases}
\]
Intuitively, player $i$ infers that each player $j$ in $\dev_i (h)$ has a type $t_j$ with $z_j(t_j) = 1$ and hence that $Z = 1$. Thus, $\bar{\a} (t) = \th$. 
\end{itemize}

Next, we define the strategy profile $\s$. Fix player $i$. Let $z_i$ denote the second-component projection from $T_i$ to $\{0,1,2\}$. We define $\s_i \colon T_i \to M_i$ by modifying $\hat{\s}_i \circ \bar{\a}_i$ after certain deviations by player $i$. Let
\[
    D_i = \{ (t_i, h) \in T_i \times \HH : \dev_i (h) = \varnothing~\text{and}~z_i (t_i) = 0~\text{and}~\bar{\a}_i(t_i) \not\in C_i(h) \}.
\]
For each $(t_i, h) \in (T_i \times \HH) \setminus  D_i$, let 
\[
    \s_i (t_i, h) = \hat{\s}_i (\bar{\a}_i (t_i), h). 
\]
 We extend $\s_i$ to $D_i$ as follows. For each $t_i \in T_i$, let $D_i (t_i) = \{ h \in \HH: (t_i, h) \in D_i \}$. Fixing the opposing strategy $\hat{\s}_{-i} \circ \bar{\a}_{-i}$, player $i$ faces a single-agent decision problem. In this problem, for each history $h$ in $D_i (t_i)$, observe that every reachable nonterminal successor of $h$ is also in $D_i(t_i)$. Therefore, given the beliefs $\mu_i$, we can define $\s_i (t_i)$ on $D_i (t_i)$ to be sequentially optimal by working backwards, starting at the histories that directly precede terminal histories, and breaking ties arbitrarily. By \eqref{eq:Bayes}, the beliefs $\mu_i$ satisfy Bayes' rule with respect to $\hat{\s}_{-i} \circ \bar{\a}_{-i}$ on $D_i (t_i)$. Thus, the one-shot deviation principle applies; see \cite{HendonEtal1996} and \cite{Perea2002}.

Our construction ensures the following \emph{path-equivalence property} between $\s$ and $\hat{\s} \circ \bar{\a}$. For each $t_i \in T_i$ and $m_{-i} \in M_{-i}$, the profiles $(\s_i (t_i), m_{-i})$ and $(\hat{\s}_i (\bar{\a}_i (t_i)), m_{-i})$ induce the same path of play. In particular, for every $\th\in\Th$,
\[
g\bigl(\s(t^\ast(\th))\bigr)
=g\bigl(\hat{\s}(\bar{\a}(t^\ast(\th)))\bigr)
=f(\a(\th)).
\]
Since $f\circ\a$ is not a selection from $F$, this decision lies outside
$F(\th)$ for some $\th\in\Th$, as required. It remains to check that $(\s, \mu)$ is a sequential equilibrium.

\paragraph{Sequential optimality} We prove that the strategy $\s$ is sequentially optimal given the belief system $\mu$. The next result allows us to reason, in certain cases, about the strategies $\hat{\s}_i \circ \bar{\alpha}_i$ rather than $\s_i$.

\begin{lem}[Sequential outcome-equivalence] \label{res:no_deviation} 
Fix a history $h \in \HH$ and a type $t_i \in T_i$. 
\begin{enumerate}[label = (\roman*), ref = \roman*]
    \item \label{it:opposing} For each $m_i \in M_i$, we have
\[
    \mu_i (t_i, h) [ g( m_i, \s_{-i} (t_{-i}); h) = g( m_i, \hat{\s}_{-i} (\bar{\a}_{-i} (t_{-i})); h)] = 1.
\]
    \item \label{it:own} If $(t_i, h) \not\in D_i$, then for each $m_{-i} \in M_{-i}$, 
\[
    g( \s_i(t_i), m_{-i}; h) = g( \hat{\s}_i (\bar{\a}_i(t_i)), m_{-i}; h).
\]
\end{enumerate}
\end{lem}

Fix player $i$. By \cref{res:no_deviation}.\ref{it:opposing}, it suffices to prove that $\s_i$ is sequentially optimal given $\mu_i$ and $\hat{\s}_{-i} \circ \bar{\a}_{-i}$.  Fix $(t_i, h) \in T_i \times \HH$. If $(t_i, h) \in D_i$, then sequential optimality is immediate, so we may assume $(t_i, h) \notin D_i$. In this case, by \cref{res:no_deviation}.\ref{it:own}, it suffices to show that $\hat{\s}_i (\bar{\a}_i(t_i))$ is sequentially optimal at $(t_i, h)$ given $\mu_i$ and $\hat{\s}_{-i} \circ \bar{\a}_{-i}$. To prove this, we separate into two cases. 
\begin{enumerate}
\item Suppose $\dev_i (h) = \varnothing$ and $\bar{\a}_i (t_i) \in C_i (h)$. Since $\hat{\s} \circ \bar{\a}$ is a BNE (\cref{res:BNE_elaboration}) and, conditional on $t_i$, the history $h$ is reached with positive probability  under $\hat{\s} \circ \bar{\a}$, it follows that $\hat{\s}_i (\bar{\a}_i(t_i))$ is sequentially optimal at $(t_i, h)$. 
\item Otherwise, by the definition of $D_i$, we must have $\dev_i (h) \neq \varnothing$ or $z_i (t_i) \neq 0$. Since $\hat{\s}_i ( \bar{\a}_i (t_i))$ is sequentially optimal at $(\bar{\a}_i (t_i), h)$ given $\hat{\mu}_i$ and $\hat{\s}_{-i}$, it suffices to prove that $\mu_i (t_i, h)$ assigns probability $1$ to the event that $\bar{\a} (t) = \th$ and $\marg_{\Th_{-i}} \mu_i (t_i, h) = \hat{\mu}_i (\bar{\a}_i (t_i), h)$.  If $\dev_i (h) \neq \varnothing$, apply \eqref{eq:mu_def}. If $\dev_i (h) = \varnothing$ and $z_i(t_i) \neq 0$, then $\pi^{\e}[ Z= 1 | t_i] = 1$, so $\pi^{\e} (\cdot | t_i)$ assigns probability $1$ to the event that $\bar{\a} (t) = \th$. Fix $\th_{-i} \in \Th_{-i}$. Let $\th = ( \bar{\a}_i (t_i),  \th_{-i})$. By
\eqref{eq:Bayes}, we have
\begin{equation*}
\begin{aligned}
    \mu_i (t_i, h) [\th] 
    &= 
    \frac{ \sum_{t_{-i} \in T_{-i}} \pi^\e ( \th, t_{-i} | t_i) \prod_{j \neq i} [ \bar{\alpha}_j (t_j) \in C_j (h)]}{\sum_{(\th', t_{-i}') \in \Th \times T_{-i}} \pi^\e ( \th', t_{-i}'| t_i) \prod_{j \neq i} [ \bar{\alpha}_j (t_j') \in C_j (h)]} \\
    &=  \frac{ \sum_{t_{-i} \in T_{-i}} \pi^\e ( \th, t_{-i} | t_i) \prod_{j \neq i} [\th_j \in C_j (h)]}{\sum_{(\th', t_{-i}') \in \Th \times T_{-i}} \pi^\e ( \th', t_{-i}'| t_i) \prod_{j \neq i} [ \th_j' \in C_j (h)]} \\
    &= \frac{  \pi^\e ( \th | t_i) \prod_{j \neq i} [\th_j \in C_j (h)]}{\sum_{\th'\in \Th } \pi^\e ( \th' | t_i) \prod_{j \neq i} [ \th_j' \in C_j (h)]} \\
    &= \frac{  \pi ( \th_{-i} | \bar{\a}_i (t_i)) \prod_{j \neq i} [\th_j \in C_j (h)]}{\sum_{\th_{-i}' \in \Th_{-i} } \pi ( \th_{-i}' | \bar{\a}_i (t_i)) \prod_{j \neq i} [ \th_j' \in C_j (h)]} \\
    &=\hat{\mu}_i ( \bar{\a}_i (t_i), h) [\th_{-i}],
\end{aligned}
\end{equation*}
where the last equality follows from the consistency of $\hat{\mu}_i$. 
\end{enumerate}

\paragraph{Consistency} To prove consistency, we construct a sequence of totally mixed strategies such that for each player $i$, the probability of trembling after observing $z_{i} \neq  1$ becomes arbitrarily small relative to the probability of any sequence of trembles after observing $z_{i} = 1$. 

Since $(\hat{\s}, \hat{\mu})$ is a sequential equilibrium, there exists a sequence, $(\hat{\s}^k, \hat{\mu}^k)$, of totally mixed behavioral strategy profiles and belief-system profiles such that for each player $i$, the following hold: (a) for each $k$, the belief system $\hat{\mu}_i^k$ is derived from $\hat{\s}_{-i}^k$ via Bayes' rule; and (b) the sequence $(\hat{\s}_i^k, \hat{\mu}_i^k)$ converges to $(\d_{\hat{\s}_i}, \hat{\mu}_i)$ pointwise on $\Th_i \times \HH$. For each $\th_i \in \Th_i$ and $h \in \HH$, let $\hat{L}_i^k (\th_i, h)$ denote the probability that type $\th_i$ makes the sequence of moves to reach $h$ under strategy $\hat{\s}_i^k$. Thus,
\[
    \hat{\mu}_i^k (\th_i, h) [ \th_{-i}] = \frac{ \pi (\th_{-i} | \th_i) \prod_{j \neq i} \hat{L}_j^k (\th_j, h)}{ \sum_{\th_{-i}' \in \Th_{-i}} \pi (\th_{-i}' | \th_i) \prod_{j \neq i} \hat{L}_j^k (\th_j', h) }.
\]
Since  $\hat{\s}_i^k \to \d_{\hat{\s}_i}$, we have \begin{equation} \label{eq:limit_equality_hat} 
    \lim_k \hat{L}_i^k (\th_i, h) = [\th_i \in C_i (h)].
\end{equation}

Next, we introduce a parameter sequence that will uniformly bound the tremble probabilities of any player $j$ with $z_j (t_j) \neq 1$. For each $k$, let
\begin{equation} \label{eq:l_k}
    \l^k = \frac{1}{k} \cdot \min_{i, \th_i, h} \hat{L}_i^k ( \th_i, h),
\end{equation}
where the minimum is over all players $i$, types $\th_i \in \Th_i$ and histories $h \in \HH$. Observe that $\l^k > 0$ because each strategy $\hat{\s}_i^k$ is totally mixed, and the minimum is over a finite set.

Define $\s_i^k$ as follows. For each $t_i \in T_i$ and $h \in \HH$, let
\[
    \s_i^k ( t_i, h) 
    =
    \begin{cases}
     \hat{\s}_i^k ( \bar{\a}_i (t_i), h)  &\text{if}~z_i (t_i) = 1, \\
     (1 - \l^k) \d_{\s_i (t_i, h)} + \l^k  \hat{\s}_i^k ( \bar{\a}_i (t_i), h)   &\text{if}~z_i (t_i) \neq 1.
    \end{cases}
\]
These mixtures are computed in $\D( M_i (h))$. By construction, $\s_i^k$ is totally mixed. Moreover, $\s_i^k \to \d_{\s_i}$. This is easily verified by cases. If $z_i (t_i) = 1$, then $D_i(t_i) = \varnothing$ so this holds because $\hat{\s}_i^k \to \d_{\hat{\s}_i}$. If $z_i (t_i) \neq 1$, this holds because $\l_k \to 0$. 

Since $\s_{-i}^k$ is totally mixed, we can define $\mu_i^k \colon T_i \times \HH \to \D( \Th \times T_{-i})$ using Bayes' rule. Let $L_i^k ( t_i, h)$ denote the probability that player $i$ with type $t_i$ makes the sequence of moves to reach $h$ under strategy $\s_i^k$. Let
\begin{equation} \label{eq:def_mu_k}
    \mu_i^k (t_i, h) [ \th, t_{-i}] = \frac{ \pi^\e (\th, t_{-i} | t_i)  \prod_{j \neq i} L_j^k (t_j, h)}{\sum_{(\th', t_{-i}') \in \Th \times T_{-i}} \pi^\e (\th', t_{-i}' | t_i) \prod_{j \neq i} L_j^k (t_j', h)}.
\end{equation}

To prove that $\mu_i^k \to \mu_i$, we record two properties of the likelihoods that follow from the definition of $\s_i^k$. For each $t_i \in T_i$ and $h \in \HH$, if  $z_i (t_i) = 1$, then 
\begin{equation} \label{eq:equality_1}
L_i^k (t_i, h) = \hat{L}_i^k (\bar{\a}_i (t_i), h).
\end{equation}
For each $t_i \in T_i$ and $h \in \HH$,  since $\s_i^k \to \d_{\s_i}$, the path-equivalence property gives
\begin{equation} \label{eq:limit_L}
    \lim_{k} L_i^k ( t_i, h) = [ \bar{\a}_i (t_i) \in C_i (h)].
\end{equation}

We need one additional result. Let $[n] = \{1, \ldots, n\}$. For each $h \in \HH$, let $\ob_i(h) = ([n] \setminus \{i\}) \setminus \dev_i(h)$. For any $t_i \in T_i$ and $h \in \HH$, let
\begin{equation*}
\begin{aligned}
    P_i (t_i) &= \{ (\th, t_{-i}) \in \Th \times T_{-i} : (\pi \otimes \k) [\th, t] > 0 \}, \\
    G_i (h) &= \Set{(\th, t_{-i}) \in \Th \times T_{-i} : z_j(t_j)  = 1~\forall j \in \dev_i (h)~\text{and}~\bar{\a}_\ell (t_\ell) \in C_{\ell} (h)~\forall \ell \in \ob_i (h) }.
\end{aligned}
\end{equation*}
Let $E_i (t_i, h) = P_i(t_i) \cap G_i(h)$. Finally, let 
 \[
    \tilde{\mu}_i^k (t_i, h) = \d_{\bar{\a}_i (t_i)} \otimes \Brac{ \hat{\mu}_{i}^k ( \bar{\a}_i(t_i), h) \otimes \k_{i}^h}.
\]

\begin{lem}[Inference after deviations] \label{res:zero_prob} For each player $i$, type $t_i \in T_i$, and history $h \in \HH$ satisfying $\dev_i (h) \neq \varnothing$, we have
\[
    \lim_k \mu_i^k ( t_i, h) ( E_i (t_i, h)) = \lim_k \tilde{\mu}_i^k(t_i, h) (E_i(t_i, h)) = 1.
\]
\end{lem}

 Fix $(t_i, h) \in T_i \times \HH$. We prove that $\mu_i^k (t_i, h) \to \mu_i (t_i, h)$. If $\dev_i (h)= \varnothing$, simply apply \eqref{eq:limit_L} to \eqref{eq:def_mu_k} and compare the result with  \eqref{eq:Bayes}. Therefore, we may assume $\dev_i(h) \neq \varnothing$. The consistency of $\hat{\mu}$ implies that $\tilde{\mu}_i^k(t_i, h) \to \mu_i (t_i, h)$. Therefore, it suffices to prove that $\| \mu_i^k (t_i, h)- \tilde{\mu}_i^k (t_i, h)\|_{\mathrm{TV}} \to 0$. We introduce notation for the numerators in the definitions of $\mu_i^k (t_i, h)$ and $\tilde{\mu}_i^k (t_i, h)$. For each $(\th, t_{-i}) \in \Th \times T_{-i}$, let
\begin{equation} \label{eq:w}
\begin{aligned}
    w_i^k (t_i, h) [ \th, t_{-i}] &= \pi^\e (\th, t_{-i} | t_i)  \prod_{j \neq i} L_j^k (t_j, h), \\
   \tilde{w}_i^k (t_i, h) [\th, t_{-i}] &= [\th_i = \bar{\a}_i (t_i)] \pi ( \th_{-i} | \th_i) \k_i^h[ t_{-i} | \th_{-i}] \prod_{j \neq i} \hat{L}_j^k (\th_j, h).
\end{aligned}
\end{equation}
We check that $w_i^k (t_i, h)$ and $\tilde{w}_i^k (t_i, h)$ are strictly positive on $E_i(t_i, h)$. The likelihoods appearing in \eqref{eq:w} are strictly positive because $\s^k$ and $\hat{\s}^k$ are totally mixed.  Since $\pi^{\e} [ Z = 1| t_i] > 0$, we have $\pi^{\e} (\th, t_{-i} | t_i) > 0$ for each $(\th, t_{-i}) \in P_i(t_i)$.  And for each $(\th, t_{-i}) \in E_i (t_i, h) = P_i(t_i) \cap G_i(h)$, we have $\bar{\a} (t) = \th$  and $\k_{i,j}^h [t_j | \th_j] > 0$ for all $j \neq i$. We show that
\begin{equation} \label{eq:limit_w_ratio}
 \max_{(\th, t_{-i}) \in E_i(t_i, h)} \Abs{ \frac{w_i^k (t_i, h) [ \th, t_{-i}]}{\tilde{w}_i^k (t_i, h) [\th, t_{-i}]} -  \frac{\pi^{\e} [ Z= 1 | t_i]}{2^{|\dev_i (h)|}}} \to 0. 
\end{equation}
This limit $\pi^{\e} [ Z= 1 | t_i]/2^{|\dev_i (h)|}$ is strictly positive and does not depend on $(\th, t_{-i})$, so we conclude that
\[
    \| \mu_i^k (t_i, h)[ \cdot | E_i(t_i, h)] - \tilde{\mu}_i^k (t_i, h)[ \cdot | E_i(t_i, h)] \|_{\mathrm{TV}} \to 0.
\]
By \cref{res:zero_prob}, it follows that $ \| \mu_i^k (t_i, h) - \tilde{\mu}_i^k (t_i, h)\|_{\mathrm{TV}} \to 0$.

Now we prove \eqref{eq:limit_w_ratio}. Since $E_i (t_i, h)$ is finite, it suffices to prove pointwise convergence over  $E_i (t_i, h)$. Fix $(\th, t_{-i}) \in E_i (t_i, h)$. The definition of $G_i(h)$ implies that $z_j (t_j) = 1$ for all $j \in \dev_i(h)$. Since $\dev_i(h) \neq \varnothing$, we have
\begin{equation} \label{eq:pi_simple}
\begin{aligned}
    \pi^{\e} ( \th, t_{-i} | t_i) 
    &= \pi^{\e} [ Z= 1 | t_i] ( \pi \otimes \k) [\th, t_{-i} | t_i]\\
    &= \pi^{\e} [ Z= 1 | t_i] \pi (\th_{-i} | \th_i) \k_{-i} (t_{-i} | \th_{-i}).
\end{aligned}
\end{equation}
Substitute in \eqref{eq:pi_simple} and apply \eqref{eq:equality_1} for each player $j \in \dev_i(h)$ to get
\[
    w_i^k (t_i, h) [ \th, t_{-i}] 
    = \pi^{\e} [ Z= 1 | t_i] \pi (\th_{-i} | \th_i) \k_{-i} ( t_{-i} | \th_{-i}) \prod_{j \in \dev_i(h)} \hat{L}_j^k (\th_j, h) \prod_{\ell \in \ob_i(h)} L_\ell^k ( t_\ell, h).
\]
Dividing by $\tilde{w}_i^k (t_i, h)[\th, t_{-i}]$, the likelihood factors $\hat{L}_j^k (\th_j, h)$ cancel for each $j \in \dev_i (h)$. For all $\ell \in \ob_i (h)$, we have
$\k_{i,\ell}^h= \k_{\ell}$. 
For all $j \in \dev_i(h)$, we have $\k_j (t_j | \th_j) = 1/2$ and $\k_{i,j}^h[t_j|\th_j] = 1$. Therefore,
\begin{equation*}
\begin{aligned}
    \frac{w_i^k (t_i, h) [ \th, t_{-i}]}{\tilde{w}_i^k (t_i, h) [\th, t_{-i}]} 
    &= \pi^{\e} [ Z= 1 | t_i] 2^{-|\dev_i (h)|} \prod_{\ell \in \ob_i (h)} \frac{L_\ell^k ( t_\ell, h)}{\hat{L}_\ell^k (\th_\ell, h)}.
\end{aligned}
\end{equation*}
For each $\ell \in \ob_i (h)$,  we have $\bar{\a}_\ell (t_\ell) = \th_\ell \in C_\ell (h)$, so $L_\ell^k (t_\ell, h) \to 1$ by \eqref{eq:limit_L} and $\hat{L}_\ell^k (\th_\ell, h) \to 1$ by \eqref{eq:limit_equality_hat}.  Hence, \eqref{eq:limit_w_ratio} follows.  

\subsection{Proof of Lemma~\ref{res:no_deviation}}

Fix $i$ and $t_i \in T_i$.  We begin by proving a stronger, sequential version of the path-equivalence property. Let $D_i(t_i) = \{ h \in \HH : (t_i, h) \in D_i \}$. Fix $h \notin D_i (t_i)$. We claim that starting from $h$, if player $i$ uses $\hat{\s}_i (\bar{\a}_i (t_i))$, then play will never reach a history in $D_i(t_i)$, no matter how their opponents play. Indeed, if $\dev_i(h) \neq \varnothing$, then $D_i (t_i)$ does not contain any successor of $h$. If $z_i (t_i) \neq 0$, then $D_i(t_i) = \varnothing$. If $\dev_i (h) = \varnothing$ and $z_i (t_i) = 0$, then $\bar{\a}_i (t_i) \in C_i( h)$ and this property is preserved as long as player $i$ follows $\hat{\s}_i ( \bar{\a}_i (t_i))$. 

From sequential path-equivalence, part~\ref{it:own} is immediate. Now we prove part~\ref{it:opposing}. Fix $h \in \HH$. For each $j$, let  $D_j(h) = \{ t_j \in T_j : (t_j, h) \in D_j\}$. By applying sequential path-equivalence to each player $j \neq i$, it suffices to show that $\mu_i (t_i, h) [ t_j \in D_j(h)] = 0$ for each $j \neq i$. Fix $j \neq i$. If $\dev_i(h) = \varnothing$, then from \eqref{eq:Bayes}, we have $\mu_i (t_i, h) [ \bar{\a}_j(t_j) \in C_j(h)]= 1$, so we are done. Therefore, we may assume $\dev_i (h) \neq \varnothing$. We separate into cases.  If $j \notin \dev_i (h)$, then $\dev_j (h) \neq \varnothing$, so $D_j(h) =\varnothing$. If $j \in \dev_i(h)$, then from \eqref{eq:mu_def}, we have $\mu_i (t_i, h)[ z_j(t_j) = 1] = 1$, so we are done.

\subsection{Proof of Lemma~\ref{res:zero_prob}}

Fix player $i$, type $t_i \in T_i$, and history $h\in\HH$ with
$\dev_i(h)\neq\varnothing$. Fix $(\th,t_{-i})\in(\Th\times T_{-i})\setminus E_i(t_i,h)$. We show that $\mu_i^k(t_i,h)[\th,t_{-i}]\to0$ and $\tilde{\mu}_i^k(t_i,h)[\th,t_{-i}]\to 0$.  Since the set $(\Th\times T_{-i})\setminus E_i(t_i,h)$ is finite, the
result then follows. 

First, suppose that there exists $j\in\dev_i(h)$ such that
$z_j(t_j)\neq1$. Since $\k_{i,j}^h = d_j^1$, we immediately get $\tilde{\mu}_i^k(t_i,h)[\th,t_{-i}]=0$ for each $k$. To analyze $\mu_i^k(t_i,h)[\th,t_{-i}]$, consider the type $t_j'=(t_j^{\bar\a_j(t_j)},1)$. Let $t_{-i}'=(t_j',t_{-\{j,i\}})$ and $\th'=\bar\a(t_i,t_{-i}')$. By \eqref{eq:equality_1}, we have $L_j^k (t_j', h) = \hat{L}_j^k (\bar{\a}_j(t_j), h)$. Since $j \in \dev_i (h)$, it follows from the path-equivalence property that player $j$ with type $t_j$ must tremble at least once to reach $h$, so $L_j^k ( t_j, h) \leq \l^k$. Therefore, 
\begin{equation} \label{eq:likelihood_bound}
    \frac{ L_j^k ( t_j, h)}{L_j^k (t_j', h)} \leq \frac{ \l_k}{\hat{L}_j^k (\bar{\a}_j(t_j), h)} \leq \frac{1}{k}.
\end{equation}
Since $\pi^{\e} [ Z= 1 | t_i] > 0$ and $\th'=\bar\a(t_i,t_{-i}')$, we have $\pi^{\e} ( \th', t_{-i}' | t_i )> 0$. By Bayes' rule, 
\begin{equation} \label{eq:k_bound}
\begin{aligned}
    \mu_i^k (t_i, h) [ \th, t_{-i}] 
    &\leq \frac{1}{k} \cdot \frac{ \pi^{\e}( \th, t_{-i} | t_i)}{\pi^{\e} ( \th', t_{-i}' | t_i)} \cdot \mu_i^k (t_i, h) [\th', t_{-i}'].
\end{aligned}
\end{equation}
Pass to the limit in $k$ to conclude that $\mu_i^k(t_i, h) [ \th, t_{-i}] \to 0$.

We may therefore suppose that $z_j(t_j)=1$ for every $j \in\dev_i(h)$. Since $\dev_i(h) \neq \varnothing$, we have $\pi^{\e} (\th, t_{-i} | t_i) = \pi^{\e} (Z = 1 | t_i) (\pi \otimes \k) ( \th, t_{-i} | t_i)$. Thus, if $(\th, t_{-i}) \notin P_i(t_i)$, we have $\mu_i^k(t_i, h) [\th, t_{-i}] = 0$ and $\tilde{\mu}_i^k(t_i,h)[\th, t_{-i}] = 0$ for all $k$. Therefore, we may assume that $(\th, t_{-i}) \in P_i(t_i)$.  Since $(\th,t_{-i})\notin E_i(t_i,h)=P_i(t_i)\cap G_i(h)$ and $z_j(t_j)=1$ for every $j \in\dev_i(h)$,   it follows that $\bar\a_\ell(t_\ell)\notin C_\ell(h)$ for some $\ell\in\ob_i(h)$. Since $\ell\in\ob_i(h)$, the set $C_\ell(h)$ is nonempty. Choose
$\th_\ell'\in C_\ell(h)$. Let $t_\ell'=(t_\ell^{\th_\ell'},1)$. Let $t_{-i}'=(t_\ell',t_{-\{\ell,i\}})$ and $\th'=\bar\a(t_i,t_{-i}')$. By construction, $(\th', t_{-i}')$ is in $P_i(t_i)$ and $t'$ agrees with $t$ over $\dev_i(h)$, so $w_i^k (t_i, h)[\th', t_{-i}']$ and $\tilde{w}_i^k (t_i, h)[\th', t_{-i}']$ are both strictly positive. We have
\[
\mu_i^k(t_i,h)[\th,t_{-i}]
=
\frac{\pi^\e(\th,t_{-i}| t_i)}
     {\pi^\e(\th',t_{-i}'| t_i)}
\frac{L_\ell^k(t_\ell,h)}
     {L_\ell^k(t_\ell',h)}
\mu_i^k(t_i,h)[\th',t_{-i}'].
\]
Since $\bar{\a}_\ell (t_\ell) \not\in C_\ell (h)$ and $\bar{\a}_\ell ( t_{\ell}') = \th_\ell' \in C_\ell (h)$,  we use \eqref{eq:limit_L} to conclude that $L_\ell^k(t_\ell,h)/L_\ell^k(t_\ell',h)$ converges to $0$. Hence, $\mu_i^k(t_i,h)[\th,t_{-i}] \to 0$. Turning to $\tilde{\mu}_i^k$, since $t_{-i}'$ and $t_{-i}$ only differ in position $\ell$, we have
\[
\tilde{\mu}_i^k(t_i,h)[\th,t_{-i}]
=
\frac{ \pi (\th_{-i} | \th_i)}{\pi ( \th_{-i}' | \th_i') } \frac{ \k_\ell (t_\ell | \th_\ell)}{\k_\ell (t_\ell'|\th_\ell')} \frac{\hat L_\ell^k(\th_\ell,h)}
     {\hat L_\ell^k(\th_\ell',h)} \tilde{\mu}_i^k(t_i,h)[\th',t_{-i}'].
\]
Since $\th_\ell = \bar{\a}_\ell (t_\ell) \notin C_\ell (h)$ and $\th_\ell' \in C_\ell (h)$, we use \eqref{eq:limit_equality_hat} to conclude that $\hat L_\ell^k(\th_\ell,h)/\hat L_\ell^k(\th_\ell',h)$ converges to $0$, so $\tilde{\mu}_i^k(t_i,h)[\th,t_{-i}] \to 0$.

\subsection{Proof of Remark~\ref{res:integer}}
 If $M_i (h)$ is countable for each $h \in \HH$ then  $\HH$ is countable. If in addition $(\HH, \HH', g)$ admits a sequential best response for each player at each history for all beliefs and opposing pure strategies,\footnote{This property is satisfied, for instance, by the integer-game mechanisms in \cite{BerginSen1998} and \cite{Baliga1999}.} the result goes through: The existence of best responses ensures that we can construct sequentially optimal strategies on $D_i$, and the  only other part of the proof where finiteness is used is in defining $\l^k$.  If $\HH$ is countably infinite, modify this step as follows. Choose a bijection $K \colon \HH \to \N$. For each $k$, let $\HH^k = \{ h \in \HH: K(h) \leq k \}$. Then let 
  \begin{equation}
    \l^k = \frac{1}{k} \cdot \min_{i \in [n], \th_i \in \Th_i, h \in \HH^k} \hat{L}_i^k ( \th_i, h).
\end{equation}  
Thus, $\l^k > 0$. For each fixed history $h$, we have $h \in \HH^{k}$ for all $k \geq K(h)$, so the argument goes through. Note, however, that the convergence may not be uniform over $T_i \times \HH$.


\subsection{Proof of Proposition~\ref{res:KPT_counterexample}}

Consider the social choice environment from \cref{sec:example}, the SCF $f$ given by $f(\th) = x^{\th}$, and the three-stage mechanism $(\HH, \HH', g)$ introduced in \cref{sec:example}. In \cref{sec:equilibria}, we proved that $(\HH, \HH', g)$ sequentially implements $f$. Now we check that $f$ is not Bayesian monotone. Consider any deception profile $\a$ such that $\a_1$ and $\a_2$ are both bijections, and at least one is different from the identity. Since $\a$ is not the identity, we have $f \circ \a \neq f$. Suppose that there exists some player $i$, payoff type $\bar{\th}_i$, and social choice function $f'$ such that $f' \circ \a \succ_{\bar{\th}_i} f \circ \a$ and, for each $\th_i \in \Th_i$, we have $f \succsim_{\th_i} f'_{i, \a_i (\bar{\th}_i)}$. In particular, taking $\th_i = \bar{\th}_i$ and using indifference over the image of $f$, we get

\[
    f' \circ \a
    \succ_{\bar{\th}_i}
    f \circ \a
    \sim_{\bar{\th}_i}
    f
    \succsim_{\bar{\th}_i}
    f'_{i,\a_i(\bar{\th}_i)}.
\]
Since $\pi$ is uniform, 
\[
 \sum_{\th_{-i} \in \Th_{-i}} u_i ( f' ( \a_i( \bar{\th}_i),  \a_{-i}(\th_{-i})), \bar{\th}_i) >   \sum_{\th_{-i} \in \Th_{-i}} u_i ( f' ( \a_i( \bar{\th}_i),  \th_{-i}), \bar{\th}_i).
\]
This is a contradiction because the two summations are the same (since $\a_{-i} (\Th_{-i}) = \Th_{-i}$).

Let $\bar{\h} = 1/20$. Let $(T, t^\ast, \hat{\pi})$ be a KPT information structure satisfying $\| \hat{\pi} - \pi \otimes t^\ast \|_{\mathrm{TV}} < \bar{\h}$. Let $\h = \| \hat{\pi} - \pi \otimes t^\ast \|_{\mathrm{TV}}$. For each $\th \in \Th$, we have $\hat{\pi}[\th, t^\ast(\th)] \geq 1/4 - \h > 0$, so type $t_i^\ast (\th_i)$ knows that their payoff type is $\th_i$. Let $(\s, \mu)$ be a sequential equilibrium. We claim that for each player $i$ and payoff type $\th_i$,  type $t_i^\ast(\th_i)$ reports $\th_i$ in the first stage. From the claim, it then follows that for each $\th_i', \th_{-i} \in \{0,1\}$, if player $i$ reports $\th_i'$ then type $t_{-i}^\ast(\th_{-i})$ of player $-i$ assigns to the event $\th_i = \th_i'$ probability strictly more than $2/3$ since
\[
    \frac{\hat{\pi} ( t^\ast (\th_i', \th_{-i}))}{\hat{\pi} (  t_{-i}^\ast (\th_{-i})) - \hat{\pi} ( t^\ast (1-\th_i', \th_{-i}))} \geq \frac{1/4 - \h}{1/4 + \h} > 2/3.
\]
Hence, player $-i$ does not challenge when they receive the distinguished signal (since  a challenged player defends if and only if their report is truthful). We conclude that $g ( \s (t^\ast(\th))) = f(\th)$ for all $\th \in \Th$.

It remains to prove the claim. Suppose that for some player $i$ and payoff type $\th_i$, type $t_i^\ast (\th_i)$ reports $1 - \th_i$ under $\s_i$. For each $\th_{-i} \in \Th_{-i}$, if player $i$ reports $1 - \th_i$, then type $t_{-i}^\ast (\th_{-i})$ of  player $-i$ assigns to the event $\th_i \neq \th_i'$ probability strictly more than $1/3$ since
\[
   \frac{ \hat{\pi} ( t^\ast (\th_i, \th_{-i}))}{\hat{\pi} ( t_{-i}^\ast (\th_{-i}))} \geq \frac{1/4 -  \h }{1/2 + \h} > 1/3. 
\]
Hence, player $-i$ challenges the report $1 - \th_i$ whenever they receive a distinguished signal. If challenged, player $i$'s payoff from reporting $1 - \th_i$ is at most $-3 + 2 = -1$. Otherwise, their payoff is at most $2$. Type $t_i^\ast(\th_i)$ believes that their opponent received a distinguished signal with probability strictly more than $2/3$ since $(1/2 - \h)/(1/2 + \h) > 2/3$. Therefore, 
their expected payoff is strictly less than $0$. But by reporting truthfully, not challenging, and defending if challenged, player $i$ guarantees a payoff of $0$, yielding a contradiction.

\subsection{Proof of Proposition~\ref{res:CI_counterexample}}

We modify the social choice environment from \cref{sec:example}. As before, there are two players, denoted $i= 1,2$. Each player $i$ observes their payoff type $\th_i \in \Th_i = \{ 0,1\}$, which is drawn uniformly and independently. But now, only player $2$ can challenge, and only when player $1$ reports $\th_1' = 1$. Formally, the set of decisions is 
\[
    A = \Set{ x^{\th', c_2, d_1}: \th' \in \Th~\text{and}~c_2, d_1 \in \{0,1\}}.
\]
Utility functions will be specified below. Let $(\HH, \HH', g)$ be the following three-stage mechanism.
\begin{enumerate}
    \item Each player $i$ reports $\th_i' \in \Th_i$. If $\th_1' = 0$, the game ends and the decision is $x^{\th'} \coloneq x^{\th', 0,0}$. Otherwise, play proceeds to stage 2. 
    \item Players observe the reports from stage $1$. Player $2$ chooses whether to challenge ($c_2 =1$) or not ($c_2 = 0$). If player $2$ does not challenge, the game ends and the decision is $x^{\th'} \coloneq x^{\th', 0,0}$. Otherwise, play proceeds to stage 3.
    \item Player 1 observes the challenge decision from stage $2$. Player $1$ chooses whether to defend ($d_1 =1 $) or not ($d_1 = 0$). The game ends and the decision is $x^{ \th', 1,d_1}$.
\end{enumerate}

The players have private values  given by the following tables:  For player $1$, 
\begin{equation} \label{eq:u_1}
\begin{array}{c|cc}
u_1(x^{\th'},0)
    & \th_2'=0 & \th_2'=1 \\ \hline
\th_1'=0 & 0 & 0 \\
\th_1'=1 & 0 & -1
\end{array}
\qquad
\begin{array}{c|cc}
u_1(x^{\th'},1)
    & \th_2'=0 & \th_2'=1 \\ \hline
\th_1'=0 & 0 & -2 \\
\th_1'=1 & -1 & 0
\end{array}.
\end{equation}
For player~2,
\begin{equation} \label{eq:u_2}
\begin{array}{c|cc}
u_2(x^{\th'},0)
    & \th_2'=0 & \th_2'=1 \\ \hline
\th_1'=0 & 0 & -1 \\
\th_1'=1 & 0 & -1
\end{array}
\qquad
\begin{array}{c|cc}
u_2(x^{\th'},1)
    & \th_2'=0 & \th_2'=1 \\ \hline
\th_1'=0 & -1 & 0 \\
\th_1'=1 & 0 & -1
\end{array}.
\end{equation}

For all $\th' \in \Th$ and $c_2, d_1 \in \{0,1 \}$, let 
\begin{equation}
\begin{aligned}
     u_1 (x^{\th', c_2, d_1}, \th_1) &= u_1 (x^{\th'}, \th_1) + c_2 v_1 ( \th_1', d_1, \th_1), \\
     u_2( x^{\th', c_2, d_1}, \th_2) &= u_2 (x^{\th'}, \th_2) + c_2 w_2 (d_1),
\end{aligned}
\end{equation}
where 
\[
    v_1 ( \th_1', d_1, \th_1) = \begin{cases}
        -1 + d_1 &\text{if}~\th_1 = \th_1', \\
        -3 - d_1 &\text{if}~\th_1 \neq \th_1',
    \end{cases}
\]
and 
\[
    w_2 (d_1) = 2 -  3 d_1.
\]
 Upon being challenged, it is strictly optimal
for player 1 to defend if they reported truthfully and not to defend if they misreported.
For player $2$, challenging changes their utility by $2$ if player $1$ does not defend but by $-1$ if player $1$ defends. Define the SCF $f \colon \Th \to A$ by $f(\th) = x^{\th}$. Define the deception profile $\a = (\a_1, \a_2)$ by $\a_1 (\th_1) = 1- \th_1$ and $\a_2( \th_2) = 0$.

To prove \cref{res:CI_counterexample}, we begin with a lemma about the set of BNE. 

\begin{lem}[BNE outcomes] \label{res:BNE} Under $\pi$, every BNE of $(\HH, \HH', g)$ induces either the SCF $f$ or $f \circ \a$.
\end{lem}

\begin{proof}[Proof of \cref{res:BNE}] Consider an arbitrary BNE $\s$. First, we check that no challenge occurs on path. Recall that player $2$ can challenge only if player $1$ reports $\th_1' = 1$. Suppose that this report is on path. If type $\th_1 = 0$ reports $\th_1' = 0$, then player $2$ concludes that the report $\th_1' = 1$ is truthful, and hence player $2$ cannot challenge this report (since player $1$ will defend). If type $\th_1 = 0$ reports $\th_1' = 1$ and player $2$ ever challenges this report, then type $\th_1 = 0$ gets utility strictly below $0$, and hence type $\th_1 = 0$ can profitably deviate to $\th_1' = 0$ to secure a payoff of $0$. 

It remains to show that, under $\s$, the players either (i) always report truthfully, or (ii) report according to $\a$. Type $\th_2 = 0$ can guarantee payoff $0$ by reporting $\th_2' = 0$ (and never challenging). Since there are no challenges on path, type $\th_2 = 0$ cannot report $\th_2 ' = 1$ since this would yield payoff $-1$.  We consider two cases according to the report of type $\th_2 = 1$. 
\begin{itemize}
    \item If type $\th_2 = 1$ reports $\th_2' = 1$, then both types of player $1$ uniquely maximize the expectation of $u_1$ by reporting truthfully, and hence must do so (since defending after reporting truthfully and being challenged yields the same payoff as not being challenged). Thus, (i) holds. 
    \item If type $\th_2 =1$ reports $\th_2' = 0$, then type $\th_1 = 1$ strictly prefers to report $\th_1' = 0$. It follows that type $\th_1 = 0$ reports $\th_1' = 1$; otherwise, type $\th_2 = 1$ could profitably deviate to the report $\th_2' = 1$. \qedhere
\end{itemize}
\end{proof}

Using \cref{res:BNE}, we now prove \cref{res:CI_counterexample} in three steps. 

1. We show that $(\HH, \HH', g)$ sequentially implements $f$. By \cref{res:BNE}, it suffices to show that $f$ can be induced by a sequential equilibrium of $(\HH, \HH', g)$, but $f \circ \a$ cannot.  For $f$,  consider the following strategy profile: both players report truthfully; player $2$ never challenges; and, when challenged, player $1$ defends if and only if they reported truthfully.  Sequential optimality is clear in stage $3$, and hence in stage $2$. For stage $1$, optimality follows from computing expected payoffs, under the uniform prior, in \eqref{eq:u_1} and \eqref{eq:u_2}. 

For $f \circ \a$, suppose for a contradiction that there exists a sequential equilibrium that induces $f \circ \a$. Consider the report history $\th' = (1,0)$. Player $2$ is certain that player $1$ misreported. By sequential optimality, player $1$ will not defend a challenge, so challenging is strictly optimal for player $2$, hence $f \circ \a$ is not induced. 

2.  We show that $f$ is not Bayesian monotone. First note that $f \circ \a \neq f$. Player $1$ cannot be a whistle-blower because each type of player $1$ gets expected utility $0$ under $f \circ \a$, and no outcome gives any type of player $1$ a strictly positive utility. For player $2$, suppose for a contradiction that there exists some payoff type $\bar{\th}_2$ and some social choice function $f'$ such that $f' \circ \a \succ_{\bar{\th}_2} f \circ \a$ and, for each $\th_2 \in \Th_2$, we have $f \succsim_{\th_2} f'_{2, \a_2 (\bar{\th}_2)}$. In particular, we take $\th_2 = \bar{\th}_2$. Since both types of player $2$ are indifferent between $f$ and $f \circ \a$, we have
\[
    f' \circ \a
    \succ_{\bar{\th}_2}
    f \circ \a
    \sim_{\bar{\th}_2}
    f
    \succsim_{\bar{\th}_2}
    f'_{2,\a_2(\bar{\th}_2)}.
\]
Since $\pi$ is uniform, plugging in the function $\a$ gives
\[
 \sum_{\th_{1} \in \Th_{1}} u_2 ( f' ( 1- \th_1, 0), \bar{\th}_2) >   \sum_{\th_{1} \in \Th_{1}} u_2 ( f' (\th_1, 0), \bar{\th}_2).
\]
This is a contradiction because the two summations are the same.

3. We complete the proof. Since the game is finite, there exists $\h_0 > 0$ such that, for any information structure $(T, t^\ast, \hat{\pi})$ satisfying $\| \hat{\pi} - \pi \otimes t^\ast \|_{\mathrm{TV}} < \h_0$, every sequential equilibrium $(\s, \mu)$ of $(\HH, \HH', g)$ is such that the strategy profile $\bar{\s}$ defined by $\bar{\s}_i (\th_i) = \s_i ( t_i^\ast (\th_i))$ is a BNE under $\pi$.  Let $\bar{\h}  = \min \{ \h_0, 1/16\}$. 

Let $(T, t^\ast, \hat{\pi})$ be a CI information structure such that $ \| \hat{\pi} - \pi \otimes t^\ast \|_{\mathrm{TV}} < \bar{\h}$. Let $\h = \| \hat{\pi} - \pi \otimes t^\ast \|_{\mathrm{TV}}$. Let $(\s, \mu)$ be a sequential equilibrium. Since $ \| \hat{\pi} - \pi \otimes t^\ast \|_{\mathrm{TV}} < \h_0$, the strategy profile $\bar{\s}$ defined by $\bar{\s}_i (\th_i) = \s_i ( t_i^\ast (\th_i))$ is a BNE under $\pi$. By \cref{res:BNE}, the SCF induced by $\bar{\s}$ is either $f$ or $f \circ \a$. Suppose for a contradiction that $\bar{\s}$ induces $f \circ \a$. Since $\bar{\s}$ is a BNE under $\pi$, it follows that $\bar{\s}_1(0)$ must specify defending, if challenged, after report profile $(1,0)$. By the sequential optimality of $\s$, type $t_1^\ast(0)$ must assign positive probability to the event $\th_1 = 1$.  We check that type $t_1^\ast (0)$ can profit by deviating to the report $\th_1' = 0$, which cannot be challenged. 
\begin{itemize}
    \item  If $\th_1 = 0$, then it can be seen from \eqref{eq:u_1} that this deviation weakly increases player $1$'s payoff. 
    \item If $\th_1 = 1$, and player $2$'s signal is either $t_2^\ast(0)$ or $t_2^\ast (1)$, then player $2$ reports $0$ and does not challenge, so by \eqref{eq:u_1}, player $1$'s gain is $1$. Player $1$'s worst gain from the deviation is $-2$. By conditional independence, the probability that player $2$ does not have either distinguished signal, conditional on each signal for player $1$ and each state, is at most $\h/(1/4 - \h) < 1/3$, where we used the fact that $\h < 1/16$. Thus, player $1$'s expected gain from deviating is strictly more than $(2/3) - 2 (1/3) = 0$. 
\end{itemize}
Since $t_1^\ast(0)$ assigns positive probability to the event $\th_1 = 1$, this deviation is strictly profitable, which is a contradiction.

\setstretch{1}

\bibliography{lit.bib}
\bibliographystyle{ecta.bst}

\end{document}